\documentclass[10pt,a4paper]{article}
\usepackage[utf8]{inputenc}
\usepackage[left=3.3cm,right=2.6cm,top=2cm,bottom=2cm]
{geometry}
\usepackage{amsmath, amsthm}
\usepackage{complexity}
\usepackage{mathtools}
\usepackage{amsfonts, dsfont}
\usepackage{amssymb}
\usepackage{physics}
\usepackage{siunitx}
\usepackage{graphicx}
\usepackage{bbold}
\usepackage{enumerate}
\usepackage{sidecap}
\usepackage [autostyle, english = american]{csquotes}
\MakeOuterQuote{"}
\usepackage{float}
\usepackage[toc,page,title]{appendix}
\usepackage{multirow}
\usepackage{braket}
\usepackage{sidecap}
\usepackage{eucal}
\usepackage{enumerate}
\usepackage{setspace}
\usepackage[super]{nth}
\usepackage{hyperref}
\usepackage{titlesec}
\usepackage{todonotes}
\usepackage[title]{appendix}
\usepackage{comment}

\usepackage[capitalise]{cleveref}
\Crefname{lemma}{Lemma}{Lemmas}
\Crefname{proposition}{Proposition}{Propositions}
\Crefname{definition}{Definition}{Definitions}
\Crefname{theorem}{Theorem}{Theorems}
\Crefname{conjecture}{Conjecture}{Conjectures}
\Crefname{corollary}{Corollary}{Corollaries}
\Crefname{example}{Example}{Examples}
\Crefname{section}{Section}{Sections}
\Crefname{appendix}{Appendix}{Appendices}
\Crefname{figure}{Fig.}{Figs.}
\Crefname{equation}{Eq.}{Eqs.}
\Crefname{table}{Table}{Tables}
\Crefname{item}{Property}{Properties}
\Crefname{remark}{Remark}{Remarks}

\DeclarePairedDelimiter\ceil{\lceil}{\rceil}
\DeclarePairedDelimiter\floor{\lfloor}{\rfloor}

\newcommand{\identity}{\mathds{1}}
\definecolor{ha}{rgb}{0.858, 0.188, 0.478}
\definecolor{tk}{RGB}{246,76,246}
\newcommand{\Htarget}{H_{\textup{target}}}
\newcommand{\Hbulk}{H_{\textup{bulk}}}
\newcommand{\Hboundary}{H_{\mathrm{boundary}}}
\newcommand{\Tpe}{T_{\mathrm{PE}}}
\newcommand{\Hpe}{H_{\mathrm{PE}}}
\newcommand{\Hsim}{H_{\mathrm{sim}}}
\newcommand{\field}{\mathds}
\newcommand{\ee}{\mathrm{e}}
\newcommand{\ii}{\mathrm{i}}
\newcommand{\tu}{t_u}

\makeatletter
\newtheorem*{rep@theorem}{\rep@title}
\newcommand{\newreptheorem}[2]{\newenvironment{rep#1}[1]{ \def\rep@title{#2 \ref{##1}} \begin{rep@theorem}}{\end{rep@theorem}}}
\makeatother

\newtheorem{theorem}{Theorem}
\newtheorem{corollary}[theorem]{Corollary}
\newreptheorem{theorem}{Theorem}
\newtheorem{lemma}[theorem]{Lemma}
\newtheorem{defn}[theorem]{Definition}

\usepackage[backend=bibtex,style=numeric,doi=false,isbn=false,url=false,maxbibnames=20,sorting=none]{biblatex}
\bibliography{pbqc}

\begin{document}

\title{Security of quantum position-verification limits Hamiltonian simulation via holography}
\author{ Harriet Apel\textsuperscript{a}, Toby Cubitt\textsuperscript{a}, Patrick Hayden\textsuperscript{b}\\Tamara Kohler\textsuperscript{c,d}, David P\'erez-Garc\'ia\textsuperscript{c}}
\date{\small\textit{\textsuperscript{a}Department of Computer Science, University College London, UK}\\  \textit{\textsuperscript{b}Stanford Institute for Theoretical Physics, Stanford University}\\ \textit{\textsuperscript{c}Instituto de Ciencias Matem\'aticas, Madrid}\\ \textit{\textsuperscript{d}Department of Computer Science, Stanford University} \vspace{-3ex}\normalsize}

\maketitle

\begin{abstract}
We investigate the link between quantum position-verification (QPV) and holography established in \cite{May} using holographic quantum error correcting codes as toy models.
By inserting the "temporal" scaling of the AdS metric by hand via the bulk Hamiltonian interaction strength, we recover a toy model with consistent causality structure.
This leads to an interesting implication between two topics in quantum information: if position-based verification is secure against attacks with small entanglement then there are new fundamental lower bounds for resources required for one Hamiltonian to simulate another.
\end{abstract}

\setcounter{tocdepth}{2}
\tableofcontents
\newpage

\section{Introduction}

Anti-de Sitter spacetime is described by the metric:
\begin{equation}\label{eqn ads metric}
ds^2 = \alpha^2 \left(-\cosh^2 \rho \, dt^2 + d\rho^2 + \sinh^2\rho \, d\Omega^2_{\mathrm{d}-2} \right),
\end{equation}
where in the limit $\rho\rightarrow \infty$ we approach the boundary of the spacetime.
Tensor network models of the AdS/CFT correspondence \cite{Happy,Random,Kohler2019,Kohler2020,Apel2021} encode the spatial component of the AdS metric by placing tensors in the cells of a honeycombing of hyperbolic space.
A free index of each tensor makes up the bulk Hilbert space while the dangling indices at the truncation of the tessellation makes up the boundary Hilbert space.
This truncation occurs at a fixed radius, $R$, that can be taken to be arbitrarily large.
These finite dimensional tensor network models capture various features of AdS/CFT while being mathematically rigorous.

These models provide an encoding map from bulk to boundary and vice versa, realising various aspects of the correspondence explicitly: complementary recovery, Ryu-Takayanagi and a duality between local models.
One could hope to explore dynamics in a semi-classical geometry by considering the dual bulk/boundary Hamiltonians associated with the toy models on space-like slices.
However, the proposed dictionaries lead to superluminal signalling on the boundary.
In \cite{Happy,Random} the boundary operations are non-local, allowing for instantaneous interactions between separate boundary regions. 
In \cite{Apel2021,Kohler2019,Kohler2020} this is improved to give a bulk-boundary mapping that maps (quasi)-local Hamiltonians in the bulk to local Hamiltonians on the boundary. 
Even so, the local boundary operations have large interaction strengths which results in unbounded signalling speeds on the boundary. 
However, from the AdS metric it is manifest that coordinate time is not spatially uniform, but instead exponentially dilates as the observer falls into the bulk.
We argue this time component of the metric needs to be fixed by hand to obtain a toy model with consistent causal structure.

By selecting a bulk Hamiltonian that incorporates this time dilation, we demonstrate that the resulting model of holography exhibits the anticipated causal structure, such that:
\begin{enumerate}
\item The time taken for a particle travelling across the bulk is consistent with the dual particle travelling around the boundary (\cref{fig: butterfly}).
\item The bulk to boundary recovery map gives dual operators that are smeared out (non-local) on the boundary but \emph{in straight light cones} giving constant butterfly velocities \footnote{Butterfly velocities give an effective maximum speed of propagation of any operator but restricted to -- in this case -- the code subspace of low energy excitation on a classical background geometry that cause negligible back reaction as in \cite{Qi2017}.} (\cref{fig: butterfly}).
\end{enumerate}
Both of these features of consistent causality were missing in previous toy models of holography where there was no constraint on the bulk Hamiltonian.

\begin{figure}[h!]
\centering
\includegraphics[trim={0cm 0cm 0cm 0cm},clip,scale=0.45]{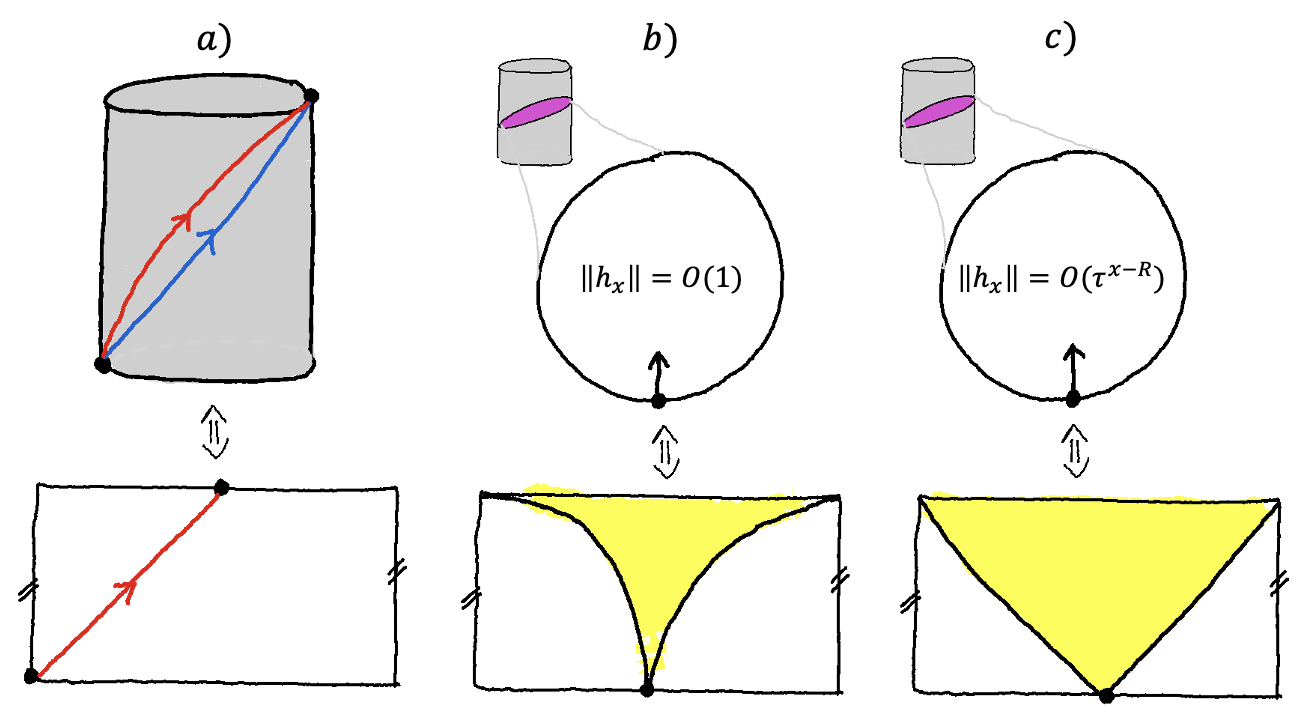}
\caption{a) With the bulk Hamiltonian norm scaling as in \cref{eqn ham scaling}, the time taken for a particle to travel through the finite bulk is consistent with the time taken for the particle to travel around the boundary in the dual theory. 
b) Ignoring the time dilation of the metric creates boundary butterfly velocities are not constant so that as a bulk object propagates to the center of the bulk the speed at which information on the boundary spreads increase. This leads to non-linear causal cones on the boundary. 
c) Accounting for time dilation recovers linear causal cones for recovered boundary operations. 
 }
\label{fig: butterfly}
\end{figure}

One motivation behind constructing tensor network toy models of AdS/CFT with the correct causality structure is to investigate quantum position-based verification (QPV).
Position verification is the cryptographic protocol which aims at securely convincing a third party of one's physical location in spacetime as a means of establishing trust. 
For example, a connection with a server could be trusted on the basis of the server being in the correct country or city.
This can replace settings typically requiring authentication with a private key, which may have been stolen or compromised. 
However, when the parties are classical, secure position verification is impossible \cite{Cha09}, so public-key schemes are the go to in modern cryptography. 
As classical public schemes are rendered insecure with the advent of quantum computers due to Shor's algorithm \cite{Ala19, Shor94}, QPV is one possible replacement in scenarios where geographic position is relevant identification information.

It has been argued that if arbitrary computations can occur in the bulk of AdS, then the bulk-boundary correspondence indicates that QPV can be broken using only linear entanglement \cite{May,May2019,dolev2022holography}. 
This is exponentially better than the best known quantum information protocols \cite{Buhrman_2014,Beigi_2011,cree2022coderouting,dolev2022nonlocal}, and there is evidence that such protocols cannot exist \cite{Junge_2022}.
The mismatch between quantum information and holography has led to speculation that gravity may constrain computation in some way \cite{May_2022}.
 
In this work we study position-based verification protocols in tensor network toy models of AdS/CFT that have the improved causality structure outlined above.
While the construction does not provide a holography-inspired QPV protocol that beats the best known quantum information protocols, it does allow us to relate two seemingly distinct fields of study in quantum information: QPV and Hamiltonian simulation. 
Hamiltonian simulation refers to the task of reproducing aspects of a target Hamiltonian up to some controllable error using qualitatively different interactions e.g. with respect to locality, geometrical structure, interaction type or symmetries.
Analogue Hamiltonian simulation is one of the most promising near-term applications of quantum technology with numerous experimental proof-of-concept experiments \cite{Porras, Jaksch2004, Peng2010, Houck2012}.
It is hoped that as quantum simulators improve they could shed light on poorly understood physics and materials.
It is known that there exist simple, nearest neighbour, geometrically local Hamiltonians on a 2D square lattice that are universal, in the sense that they can simulate \emph{any} Hamiltonian \cite{Cubitt2019}, albeit with the caveat that this requires the simulator system to have very large interaction strengths. 
This control of interaction strengths over many orders of magnitude is beyond the reach of current technology, so known families of universal Hamiltonians are not practical simulators.
However, the current constructions of universal Hamiltonians may not be optimal.
General bounds on the interactions strengths required for one Hamiltonian to simulate another provides insight into where constructive simulation techniques may be optimised to broaden the scope of analogue simulators. 

We are able to show that if QPV is secure against attacks with linear entanglement then this implies fundamental limits on Hamiltonian simulation techniques:
\begin{reptheorem}{them PBQC sim}[QPV and simulation (informal)]
If quantum position-verification is secure against attacks with linear entanglement, then for any universal family of 2-local Hamiltonians $\mathcal{M}$~\footnote{We say that a family $\mathcal{M}$ of Hamiltonians is universal if for {\it any} target Hamiltonian $\Htarget$ there exists an $H \in \mathcal{M}$ that can simulate $\Htarget$.} there exists some target Hamiltonian $\Htarget$ acting on $N$ spins such that the local interaction strength required to simulate $\Htarget$ with a Hamiltonian from $\mathcal{M}$ is lower bounded by a polynomial in the ratio of $N$ to the desired error in the simulation. 
 \end{reptheorem} 
 
See the formal statement of \cref{them PBQC sim} along with \cref{def: attacks} and \cref{def: very good simulators},  for a rigorous statement of the result, which outlines the lower bounds in more detail.

It has previously been shown in \cite{aharonov_2018} that there exist 2-local, all-to-all interactions which cannot be simulated by 2-local interactions with constant degree using $O(1)$ interaction strength.
If the security of QPV against attacks with linear entanglement was demonstrated then \cref{them PBQC sim} would give new lower bounds on achievable simulations.  
Note that in our result by linear entanglement we mean the amount of entanglement needed to break a QPV scheme on $n$ qubits is upper bounded by $C_1n$ where we have $C_1 > 1$.
$N$, the number of qubits that $\Htarget$ acts on in \cref{them PBQC sim}, is proportional to $n$, the number of qubits on which the local unitary in the position based protocol acts on.
Since the constant $C_1$ is greater than 1, it is not known whether QPV is secure against attacks of this type as we are outside the regime in \cite{Tomamichel2013} where it is known to be insecure\footnote{It is known that protocols exist using $n$ qubits that are secure against adversaries with at most $C_1 n$ entanglement where $C_1<-\log_2\left(\frac{1}{2}+\frac{1}{2\sqrt{2}} \right)$ \cite{Tomamichel2013}. If qudits are used this constant can be improved by using monogamy games.}.

There are currently no known Hamiltonian simulation methods that achieve the interaction strength scaling in \cref{them PBQC sim}.
Our argument  implies that the existence of such techniques would provide a protocol to break position-based verification with only linear entanglement.
Note that the converse is not true -- it may be possible to break position-based verification via other methods that say nothing about fundamental limits of simulation. 

While \cref{them PBQC sim} is stated in terms of general unitaries and general Hamiltonians, there is an immediate corollary about the security of families of unitaries and the complexity of simulating the corresponding Hamiltonians.
The corollary says that Hamiltonians corresponding to unitaries that are secure for QPV are hard to simulate.\footnote{Note that the hard-to-simulate Hamiltonian that corresponds to the unitary is not quite the Hamiltonian that generates the unitary -- it is related to this Hamiltonian by an isometric mapping -- see \cref{them PBQC sim entire physics} for details.}
Imposing no structure on the bulk unitary, and using \cref{them PBQC sim} to bound the resources to simulate global Hamiltonians, gives a no-go simulation result in line with what could be expected from physical intuition: simulating a highly non-local Hamiltonian by a local Hamiltonian will require strong interactions on the simulating system. 
However, when considering restricted classes of unitaries the bound implied by holography becomes more interesting. 
For example, non-trivial unitaries for QPV can be generated by $k$-local bulk Hamiltonians which leads to  sparse target Hamiltonian.
In this case by optimising existing simulation techniques we can in fact come very close to saturating the bound from \cref{them PBQC sim}.
Saturating the bound would imply that unitaries generated by $k$-local bulk Hamiltonians cannot be used in QPV protocols that are secure against linear entanglement. 
On the other hand, showing that such unitaries can provide secure QPV protocols, would imply that the simulation technique derived here is close to optimal for this class of Hamiltonians.  
Previous results have shown that circuits that have low complexity as measured by their $T$-gate count \cite{speelman:2016} do not provide secure unitaries for QPV.\footnote{In \cite{speelman:2016} it is demonstrated QPV protocols using circuits with logarithmic numbers of $T$ gates can be broken using only polynomial entanglement.}
Our result can be viewed as another way of exploring this relationship between security for QPV and complexity. 

The rest of this paper is structured as follows. 
In \cref{sect: Causal struct} we describe how the time dilation described above can be fixed by hand into the tensor network models to yield the features of consistent causal structure described above.
In \cref{sect: pbqc} we then use these models to prove \cref{them PBQC sim} for a strong definition of simulation where we have a relationship between simulation parameters.
We then compare the bounds on simulation in \cref{them PBQC sim} to the optimised version of best-known simulation techniques for sparse Hamiltonians.

\section{Causal structure of toy models of AdS/CFT from time dilation}\label{sect: Causal struct}

\subsection{Time dilation in the bulk}

The bulk of the holographic duality is described by the $\text{AdS}$ metric.
In this metric time for an observer in the centre of the bulk is dilated with respect to the time experienced by an observer sitting on the conformal boundary.
This can be seen by looking at the $\text{AdS}_{\mathrm{d}}$ metric in global coordinates
\begin{equation}\label{eqn ads metric 2}
ds^2 = \alpha^2 \left(-\cosh^2 \rho \, dt^2 + d\rho^2 + \sinh^2\rho \,  d\Omega^2_{\mathrm{d}-2} \right)
\end{equation}
where $\rho \in \mathbb{R}^+$ and in the limit $\rho\rightarrow \infty$ we approach the boundary of the spacetime.
$\alpha$ is the radius of curvature, related to the cosmological constant by $\Lambda = \frac{-(\mathrm{d}-1)(\mathrm{d}-2)}{2\alpha^2}$.
At fixed spatial position ($\rho, \mathbf{\Omega}$) the proper time $d\mathcal{T} = - ds$ is given by,
\begin{equation}
d\mathcal{T} = \alpha \cosh \rho \,dt.
\end{equation}

This "time dilation" can be observed by comparing the coordinate time for two observers at different distances into the bulk.
Consider one observer at $\rho_0$ and one at $\rho_1$ with $0<\rho_0<\rho_1$.
The relationship between the two coordinate times with the same proper time is,
\begin{equation}
dt_0  = \frac{\cosh \rho_1}{\cosh \rho_0} dt_1.
\end{equation}
Note that $dt_0 > dt_1$ since for positive arguments $\cosh$ is monotonically increasing $\frac{\cosh \rho_1}{\cosh \rho_0}>1$.

\subsection{"Fixing" time dilation in the toy model}

The spatial component of the AdS metric is incorporated into the model by placing the tensors in the cells of a honeycombing of hyperbolic space.
In order to emulate the time component of the metric we need to insert this dilation by hand.
Since this is a finite model we consider a truncated version of AdS spacetime where instead of infinity the conformal boundary exists at a fixed maximum radius, $R$, that can be taken to be arbitrarily large.
This truncation fixes a central tensor that represents the deepest part of the bulk.  
In translating the metric onto the tensor network model, the distance from the centre of the bulk is now quantified by the number of layers of tensors from this central tensor, which we denote $x\in \mathbb{Z}^+$ (see \cref{fg bulk radius}).

In constructing a tensor network toy model there are a number of free parameters: the type of tensor, the radius of truncation $R$, the local Hilbert space dimension $d_b$ of each dangling leg associated with the bulk degrees of freedom and the local Hilbert space dimension $d_\partial$ of the tensor legs contracted in the bulk and associated with the boundary degrees of freedom.
Previous tensor network models \cite{Happy,Kohler2019} use perfect tensors\footnote{Perfect tensors correspond to absolutely maximally entangled states (see Definition 2 of \cite{Happy}) and have the property that the tensor is an isometric map across any bipartition of legs.} for the construction whereas other constructions use high-dimensional random tensors\footnote{Random tensors correspond to Haar random states (see Section 2.2 of \cite{Apel2021}) and in the limit of large bond dimension are perfect tensors with high probability.} \cite{Random,Apel2021} -- both of which have desirable bulk-boundary reconstruction and entanglement properties that mirror the true AdS/CFT correspondence.
For convenience, in later sections we chose to make each tensor `leg' a bundle of qubits rather than a singular qudit. 
Let $n$ denote the number of qubits in each bulk leg, and $m$ denote the number of qubits in each leg which is contracted in the bulk or left uncontracted on the boundary, as in \cref{fg bulk radius}.
In constructions using perfect tensors $n=m$, whereas in random tensor constructions it is essential for the desirable features of the models to take $m\gg n$.

\begin{figure}[h!]
\centering
\includegraphics[trim={0cm 0cm 0cm 0cm},clip,scale=0.3]{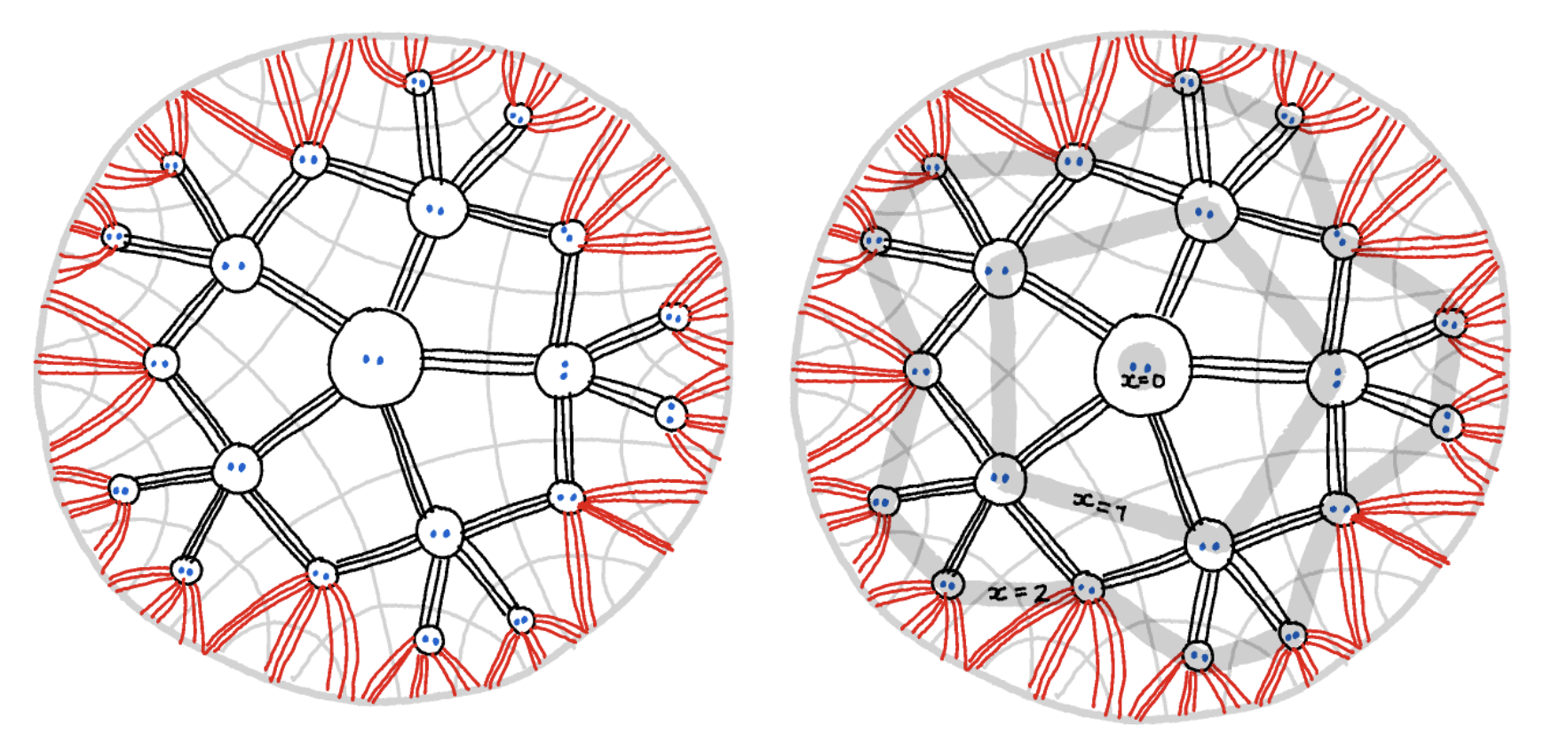}
\caption{Left) tensor network with $n=2$ and $m=3$. Right) Bulk of AdS. $x$ denotes the position from the centre and the conformal boundary is at $R=2$.}
\label{fg bulk radius}
\end{figure}

There is a choice when relating the natural "radius" in the tensor network model to the measure of radius in the metric, $\rho$.
The relationship between the value of $\rho$ at position $x$, denoted $\rho_x$, and $x$ must be consistent with how the curvature of AdS space causes the length swept out by the angle coordinate to grow as the radius increases. 
In the tensor network picture the surface area at $x$ is given by the number of degrees of freedom in that layer $N(x)$. 
If each tensor has $n$ dangling bulk indices, the length then corresponds to $n$ times the number of tensors in the layer, which grows with $x$ exponentially giving:
\begin{equation}\label{eqn N(x)}
N(x)\propto n\tau^x,
\end{equation}
where $\tau$ is a constant associated with the hyperbolic honeycombing\footnote{The constant $\tau$ is the growth rate of the Coxeter group associated with the honeycombing -- see section 5.2 of \cite{Kohler2019} and definitions within.}.
To be consistent $N(x)$ should correspond to the circumference of a circle with radius $\rho_x$, $c(\rho_x)$, as taken from the metric in \cref{eqn ads metric}:
\begin{equation}\label{eqn c(px)}
c(\rho_x) = \int_0^{2\pi} \sinh \rho_x d\Omega = 2\pi \sinh\rho \sim e^\rho.
\end{equation}
Equating \cref{eqn N(x)} and \cref{eqn c(px)} gives the correspondence between $x$ and $\rho_x$:
\begin{align}
N(x) &= c(\rho_x)\\
n\tau^x & = e^{\rho_x}\\
x \ln \tau + \ln n &= \rho_x.
\end{align}

Therefore, the coordinate time at the truncated boundary, $t_R$, is related to the coordinate time in the bulk at the $x$'th tensor layer by,
\begin{align}
dt_x  &= \frac{\cosh \rho_R}{\cosh \rho_x} dt_R\\
& = \frac{e^{\ln(n\tau^R)}+ e^{-\ln(n\tau^R)}}{e^{\ln(n\tau^x)}+ e^{-\ln(n\tau^x)}} dt_R\\
& \sim \tau^{R-x} dt_R.\label{eqn time scaling}
\end{align}
Note the scaling is independent of $n$ and $m$.

In holography the entire physics of the bulk has a representation on the boundary, hence to consider a boundary theory associated with a boundary coordinate time it is more convenient to put this time dilation into the bulk Hamiltonian as a renormalisation.
This is mathematically equivalent since the Hamiltonian and time always appear in a product. 
Assume we have a geometrically 2-local Hamiltonian in the bulk where interacting terms only occur between tensors in neighbouring tessellation cells:
\begin{equation}
H_\textup{bulk} = \sum_{i} h^{(i)}_x,
\end{equation}
where $h_x$ acts trivially on all tensors in layers $y<x$ of the tessellation (since the Hamiltonian is 2-local it will also act trivially on the majority of tensors in layers $\geq x$ but for time dilation this is less relevant).

To encode the time dilation scaling in \cref{eqn time scaling} the norm of the bulk interaction terms should decay exponentially towards the centre of the bulk:
\begin{equation}\label{eqn ham scaling}
\norm{h_x} = O\left(\tau^{x-R} \right).
\end{equation}
The Hamiltonian terms acting at $x=R$ are $O(1)$ which is consistent since here the bulk and boundary proper time coincide. 

\subsection{Rescaled Hamiltonians in AdS}

We have argued for the above scaling via time dilation, however this rescaling is also manifest looking at Hamiltonians in AdS.
For example this exponential scaling is seen in the Hamiltonian of a free particle moving in AdS: $H_\textup{free particle} = \int d\rho \, d \phi \, \cosh \rho \left(\tanh \rho \, d \psi_0 \right)^2$.
More generally recall that on a background described by metric $g_{\mu \nu}$, the Hamiltonian takes the form 
\begin{equation}
H = \int d^{d-1}x\sqrt{h}\mathcal{H}.
\end{equation}
For $\sqrt{h}$ the metric determinant (of the spatial part of the metric only) and $\mathcal{H}$ the Hamiltonian density.
For $\text{AdS}_{2+1}$, the metric determinant in these coordinates is $\cosh(\rho)$, this gives the exponential scaling of the interaction strength.

\subsection{Sending a particle through the bulk/around the boundary}\label{sect: send through}

We have argued that choosing a bulk Hamiltonian that scales as \cref{eqn ham scaling} is necessary to have causal structure in the bulk/boundary in a toy model of holography consistent with the AdS metric.
It is notable that a similar exponential envelope on the interaction strengths is also exactly what is used in \cite{Bettaque} to obtain the correct specific heat capacity of the SYK model from a MERA-like network.
Consistent causal structure is manifest when we consider a particle travelling from one side of the boundary to the other as in \cref{fg through bulk}.

\begin{figure}[h!]
\centering
\includegraphics[trim={0cm 1.5cm 0cm 1.5cm},clip,scale=0.3]{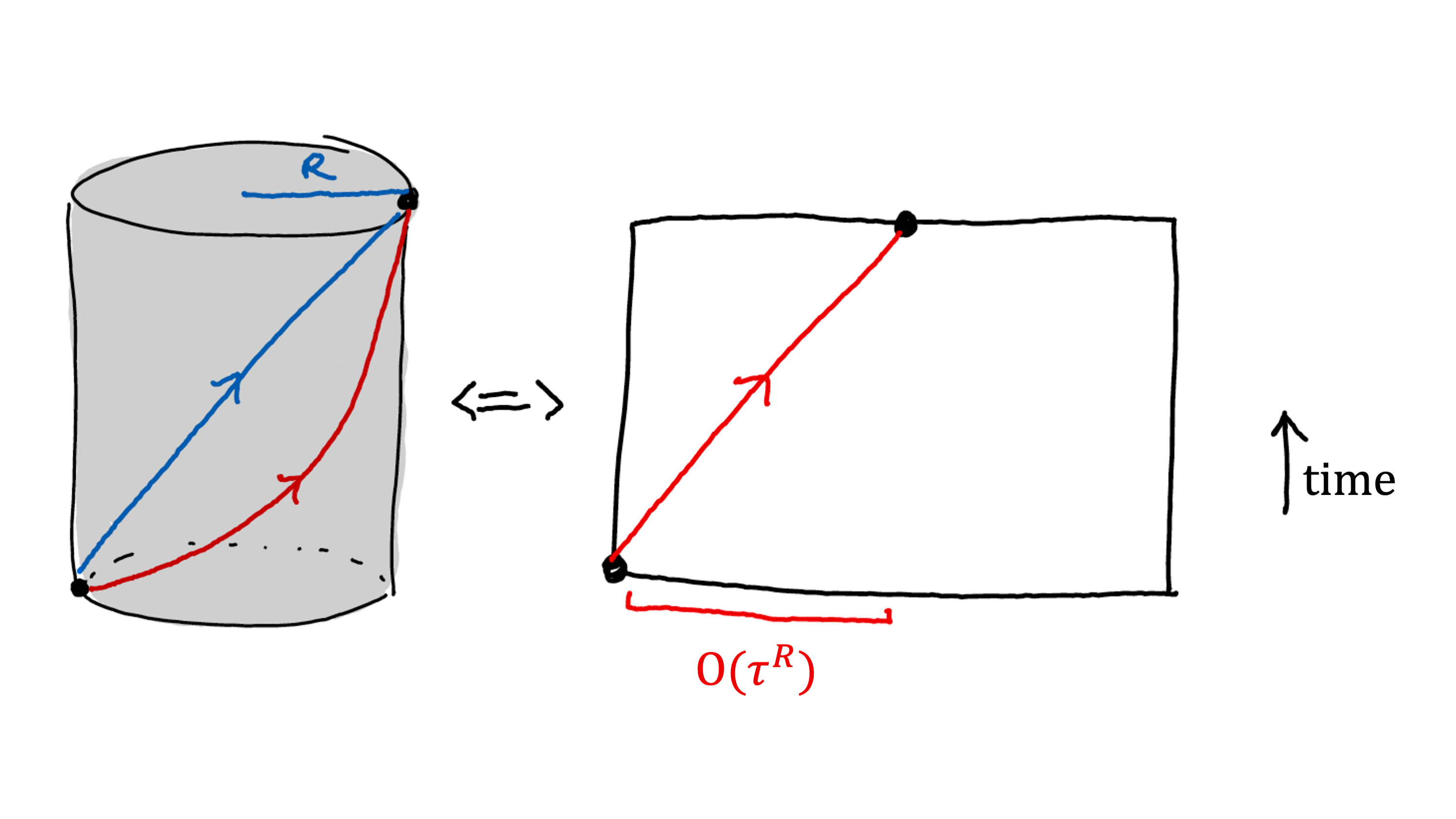}
\caption{Sending a particle sent through the bulk/around the boundary.}
\label{fg through bulk}
\end{figure}

In the bulk theory the particle falls in to the bulk from the asymptotic boundary and eventually reaches the other asymptotic boundary evolving under the bulk Hamiltonian.
In the boundary theory the particle travels around the boundary. 
On the boundary the particle has to traverse $O(\tau^R)$ tensors, whereas in the bulk it only has to make $O(R)$ `jumps'.
If the bulk and boundary Hamiltonian were both 2-local and had uniform norm then the two theories would take different times to achieve the same outcome breaking the consistency of the holographic picture.

Under a Hamiltonian with interaction scaling as \cref{eqn ham scaling}, it takes the falling particle $O(n\tau^{R-x})$ time to travel from tensor layer $(x+1)$ to $x$, an $O(n)$ distance. 
Hence the particle takes time
\begin{equation}
T_1 = 2n\sum_{x=0}^R \tau^{R-x} = \frac{2n(\tau^{R+1}-1)}{\tau -1} = O(n\tau^R),
\end{equation}
to travel through the bulk. 
Assuming a geometrically 2-local boundary Hamiltonian with all interactions scaling as $\norm{h_R} = O(1)$ to be consistent with \cref{eqn ham scaling}, the particle takes time
\begin{equation}
T_2 = \sum_{i=0}^{O(m\tau^R)} 1 = O(m\tau^R),
\end{equation}
to travel around the boundary. 

The scalings of $T_1$ and $T_2$ match up with respect to $R$, hence by choosing an appropriate bulk Hamiltonian we have restored a consistent causal structure for this basic task. 
If $m\neq n$ then the boundary Hamiltonian must scale as the ratio $\norm{h_\textup{boundary}} = O(m/n)$ to recover $T_1 = T_2$.
It is important that the exponential decay of the bulk Hamiltonian norm have the exponent given in \cref{eqn ham scaling}.
Having even a different constant factor in the exponent would cause the above bulk/boundary tasks to take different times and therefore lose the connection to holography.

\subsection{Causal behaviour of perturbations from butterfly velocities}

Another motivation for rescaling the bulk Hamiltonian is that when considering a bulk perturbation, the dual boundary causal cone is now straight, with constant velocity regardless of where in the bulk the perturbation is. 
In previous models with uniform strength bulk Hamiltonians the information would appear to propagate faster on the boundary the deeper the dual excitation is in the bulk.
The rescaling makes sense of this peculiarity.

We will use the bulk Lieb-Robinson velocity to determine the butterfly velocity via the duality.
In the bulk there is a geometrically 2-local Hamiltonian that respects the metric of the bulk.
\begin{defn}[$k$-local Hamiltonian]\label{defn local Hamiltonian}
Given a Hamiltonian $H = \sum_j h_j$ on $n$ qudits $(\mathbb{C}^d)^{\otimes n}$.
We say that $H$ is $k$-local if each $h_j$ acts `non-trivially' on a subset $S_j$ of at most $k$ qubits and as the identity everywhere else.
I.e.
\[\forall j:\qquad h_j = h_{S_j}\otimes \mathbb{1}_{n\setminus S_j} \quad \textup{and} \quad \abs{S_j}\leq k.\]
\end{defn}
\noindent A local Hamiltonian introduces the concept of a interaction graph and hence an interaction distance between two regions.
\begin{defn}[Interaction distance]\label{defn Interaction distance}
Given a $k$-local Hamiltonian $H = \sum_{S_j}h_j$  and two subsets of qubits $X$ and $Y$ the interaction distance 
\[d(X,Y):= \min \abs{\left\{ S_j : X\cap S_j,S_{j}\cap S_{j+1},S_n\cap Y \neq \empty \right\}_{j=1,..,n}}\]
is the minimum number of steps alongs interactions to get from $X$ to $Y$.
\end{defn}
\noindent The Lieb-Robinson velocity is the maximum speed that information can propagate in a theory, formally in the Heisenberg picture
\begin{theorem}[Lieb-Robinson velocity \cite{LRbound}]\label{defn LR}
Given a $k$-local Hamiltonian, $H= \sum_Z h_Z$ such that $\exists$ $\mu,s>0$ where for all qudits $i$: $\sum_{Z\ni i}\norm{h_Z}\leq s e^{-\mu}$, consider $A_X$ and $B_Y$ are operators acting on subsets of qudits $X$, $Y$.
Then, 
\begin{equation}
\norm{\left[A_X(t),B_Y \right]}\leq \min \{ \abs{X},\abs{Y}\} \norm{A}\norm{B}e^{-\mu(d(X,Y)-vt},
\end{equation}
where $d(X,Y)$ is the interaction distance between sets $X$ and $Y$ and $v$ is the \emph{Lieb-Robinson velocity}.
For such a local Hamiltonian $v= \frac{2ks}{\mu}$.
\end{theorem}

While Lieb-Robinson speeds bound the propagation of any operator, butterfly velocities give an effective maximum speed of propagation of any operator \emph{restricted to a given subspace} of the Hilbert space.
In \cite{Qi2017} they study butterfly velocities in the context of holography where the subspace in question is the code subspace of low energy excitation on a classical background geometry that cause negligible back reaction. 
In the tensor network picture, these holographic butterfly velocities correspond to restricting the boundary to the codespace associated with a fixed tensor geometry with a tensor present in each cell of the tessellation\footnote{See \cite{Kohler2019} section 3.2 for a discussion of how considering a superposition of multiple networks with different tensors removed could encode dynamical geometry.}. Formally,

\begin{defn}[Butterfly velocity $v(O_A, \mathcal{H}_c)$; \cite{Qi2017}]\label{dfn butterfy}
Given a boundary code subspace, $\mathcal{H}_c$ and a boundary operator supported on boundary region $A$, $O_A$, the butterfly velocity, $v(O_A,\mathcal{H}_c)$ is the minimum velocity such that $\forall$ $\epsilon>0$ $\exists$ $\delta >0$ where $\Delta t < \delta$ and the following two statements are satisfied:
\begin{enumerate}
\item If $D>\left(v(O_A,\mathcal{H}_c) +\epsilon\right)\Delta t$ then \\
$\forall$ $B\in \{R|d(R,A)=D \}$\\
$\forall$ $O_B$ supported in boundary region $B$\\
$\forall$ $\ket{\psi_i},\ket{\psi_j}\in \mathcal{H}_c$:
\begin{equation}
\bra{\psi_i} \left[O_A, O_B \right] \ket{\psi_j} = 0.
\end{equation} 
\item If $D<\left(v(O_A,\mathcal{H}_c) -\epsilon\right)\Delta t$ then \\
$\exists$ $C\in \{R|d(R,A)=D \}$\\
$\exists$ $O_C$ supported in boundary region $C$\\
$\exists$ $\ket{\psi_i},\ket{\psi_j}\in \mathcal{H}_c$:
\begin{equation}
\bra{\psi_i} \left[O_A, O_C \right] \ket{\psi_j} \neq 0.
\end{equation} 
\end{enumerate}
Where $d(A,B)$ is the interaction distance.
\end{defn}

\begin{figure}[h!]
\centering
\includegraphics[trim={0cm 0cm 0cm 0cm},clip,scale=0.36]{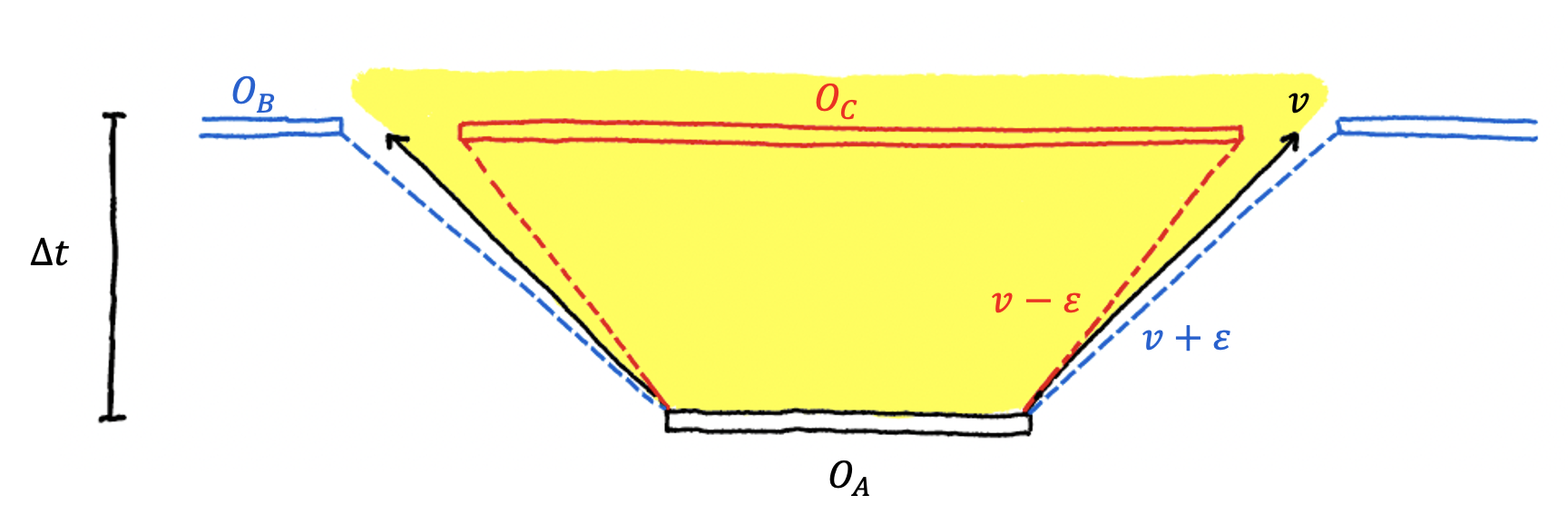}
\caption{Depiction of the definition of butterfly velocity given in \cref{dfn butterfy}. }
\label{fg butterfly defn}
\end{figure}

Section 4 of \cite{Qi2017} describes how a bulk theory with a known speed of light determines the boundary butterfly velocity.
See \cref{fg butterfly protocol} for a summary.

\begin{figure}[h!]
\centering
\includegraphics[trim={0cm 0cm 0cm 0cm},clip,scale=0.35]{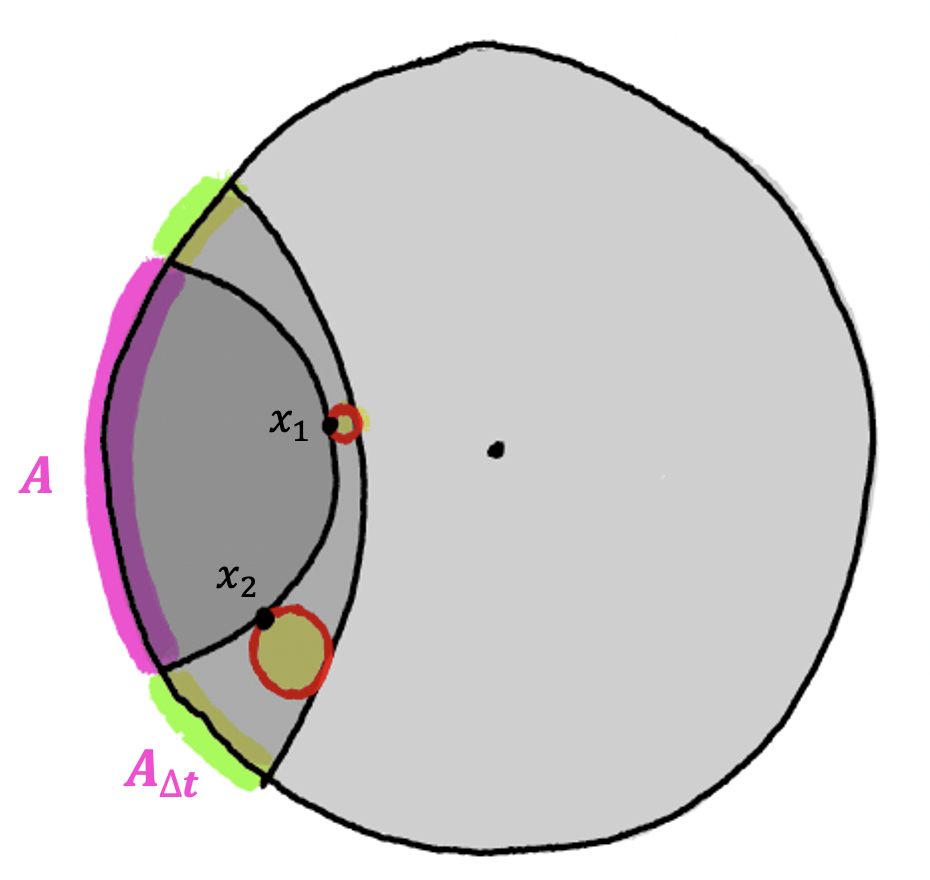}
\caption{Bulk Lieb-Robinson speed gives boundary butterfly velocity. In AdS/CFT a boundary subregion $A$ can recover any bulk operator that lies within its entanglement wedge, defined by the extremal surface $\gamma_A$ homologous to $A$. The same subregion $A$ contains no information about bulk operators outside of this entanglement wedge. 
At boundary time $t=0$ consider a boundary subregion $A$ and its extremal surface $\gamma_A$. Take any bulk point $x$ infinitesimally inside the entanglement wedge from $\gamma_A$. Any bulk operator supported at $x$, $\phi_x$, can be reconstructed on the boundary subregion $A$, call the resulting boundary operator $O_A[\phi_x]$. The butterfly velocity $v(O_A[\phi_x], \mathcal{H}_c)$ for the codespace $\mathcal{H}_c$ is equal to the minimal velocity $v^*$ such that the extremal surface bounding the expanded region $\{A_{\Delta t} | d(A,A_{\Delta t}) = v^* \Delta t)\}$ is tangential to the light-cone of $x$ in the bulk defined by the bulk Lieb-Robinson velocity.
This figure depicts that if the bulk operator is picked at a smaller radius $x_1$ along the minimal geodesic than another $x_2$, then this formalism implies that the light cone spread at $x_1$ must be smaller than at $x_2$ in order to coincide with the minimal geodesic of $A_{\Delta t}$. }
\label{fg butterfly protocol}
\end{figure}

The complementary recovery feature necessary for the butterfly velocity argument in \cite{Qi2017} is exhibited in many tensor network models of AdS/CFT, for example \cite{Happy, Random, Kohler2019, Kohler2020, Apel2021}\footnote{For the networks comprised of perfect tensors there are some pathological boundary subregions where complementary recovery is violated, however it was shown in \cite{Apel2021} with a network of random stabilizer tensors and special tessellations that this holds for any boundary subregion.}.
For a geometrically 2-local bulk Hamiltonian with individual interactions bounded in norm by \cref{eqn ham scaling} a bulk Lieb-Robinson velocity is given by
\begin{equation}
v_\textup{bulk}(x) = 2 \norm{h_x} = O(\tau^{x-R}).
\end{equation}
Hence the bulk Lieb-Robinson velocity varies with the depth into the bulk -- a feature of inserting by hand the dilation of proper time. 

Consider a connected subregion of the boundary, $A_t$ where the deepest the entanglement wedge extends into the bulk is to tensor layer $(x+1)$.
Complementary recovery and the geometry of the tessellation dictates the subregion will contain $O(m\tau^{R-x-1})$ boundary qubits. 
By the same argument the augmented boundary subregion $A_{t+\Delta t}$ that now extends deeper into the  bulk to the $x$'th layer contains $O(m\tau^{R-x})$ boundary qubits. 
By butterfly velocities, the causal boundary subregion has increased by $(\tau - 1)m\tau^{R-x-1} = O(m\tau^{R-x})$ boundary qubits while the bulk particle has moved from layer $x+1$ to layer $x$.
Now consider the lightcone of the edge of the entanglement wedge of $A_t$.
By choosing a bulk Hamiltonian to scale as in \cref{eqn ham scaling} the boundary-coordinate-time for the lightray to reach the $x$'th layer from the $(x+1)$'th layer is $O(n\tau^{R-x})$.
Therefore the boundary butterfly velocity is $O(m/n)$ crucially independent of $R$ and $x$.

Therefore via butterfly velocities this section has argued that given any holographic tensor network model that exhibits complementary recovery, if a local Hamiltonian with exponentially decaying interaction strength is placed in the bulk, then local bulk operators are smeared out non locally in the boundary in straight light cones.
This implies a local boundary Hamiltonian that also has a fixed speed of light as expected in CFTs.
Restricting the bulk Hamiltonian to those that reflect the time dilation of the AdS metric is necessary for this aspect of the causal structure of the correspondence to be observed in toy models.

\section{Secure finite dimensional QPV implies an upper bound on simulation techniques}\label{sect: pbqc}

\subsection{Quantum position-verification}

A key motivation for exploring the causal structure of holographic toy models is the potential of holographic protocols to provide insight into the security of quantum position-verification (QPV).

The fundamental idea behind QPV is that the physical location of a party serves as a unique identifier and can be used to establish trust.
Unlike traditional cryptographic schemes that rely on keys and passwords, the authentication factor here is position in space-time. 
The honest prover publicly claims to control the space-time region and is allowed to send and receive quantum messages in all directions and perform quantum computation.
There are multiple third-party referees that together look to gain cryptographic evidence of the prover's location.
They send challenges to the prover in the form of quantum states and the prover should respond to the inputs by applying a unitary $U$ and sending answers back to the referees in a set time.
$U$ is part of the protocol and thus public knowledge.
The size of the protocol is determined by the total size of the challenge, $n$ qubits.
$U$ is in general a $n$ qubit unitary. 
The protocol set-up and attack model are demonstrated in \cref{fig:circuitdiag}.

\begin{figure}[h!]
\centering
\includegraphics[trim={0cm 0cm 0cm 0cm},clip,scale=0.4]{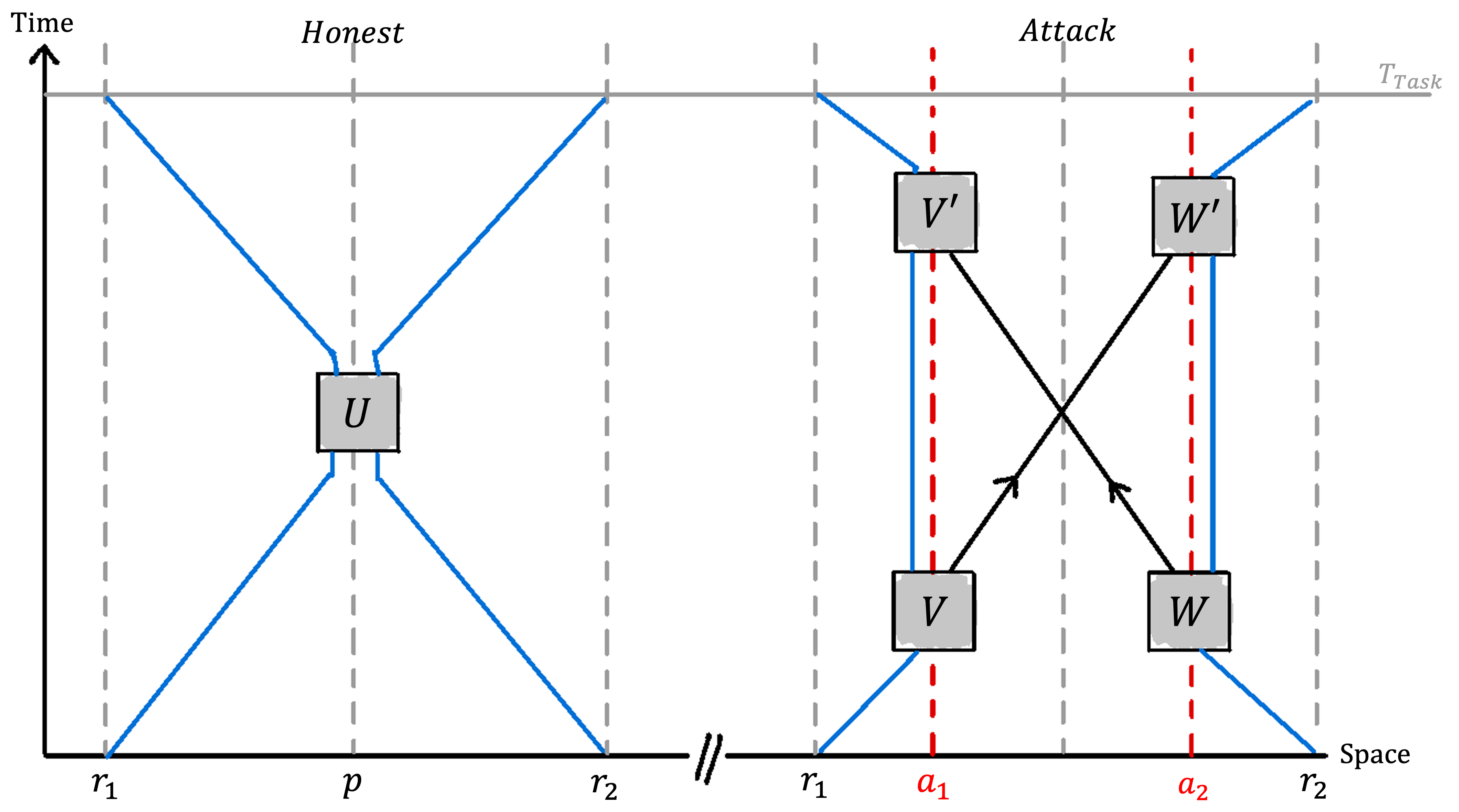}
\caption{Circuit diagram for the honest protocol and attack model for quantum position-verification in the (1+1) dimensional case. \textbf{Left}: Honest protocol. The honest prover claims to be at spatial position $p$ and is surrounded by two third-party referees $r_1$, $r_2$ working together to test this claim. They each send to position $p$ a quantum state, in total of $n$ qubits. If the prover is at $p$ they receive these inputs and apply a local $n$-qubit unitary $U$ and send half the output to each referee which they receive at $T_\text{Task}$. If the outputs reflect the expected outcome of applying the $U$ to the inputs then the referees accept this as cryptographic evidence of the provers location. \textbf{Right}: Instead of an honest prover in the attack model there are two adversaries $a_1$, $a_2$ outside of the position $p$. They each intercept an input and implement the first round of local unitaries $V$, $W$. They are allowed one round of simultaneous communication. They then perform a second round of local unitaries $V'$, $W'$ on their part of the challenge before sending the output to respective referees by $T_\text{Task}$. If the referees are again satisfied the answer reflects successful application of $U$ they accept and the position-verification has been rendered insecure. The two adversaries are allowed to pre-share an entangled state. The question of security of QPV then reduces to how amount of the pre-shared entanglement required to break the protocol scales with the size of the protocol ($n$).}
\label{fig:circuitdiag}
\end{figure}

A group of adversaries outside of the specified region would compromise the security of the QPV protocol if they could spoof the verifiers into believing they are in the prescribed location.
It is known that it if the adversaries are allowed to share arbitrary entangled states then they can break any QPV protocol \cite{Buhrman_2014}.
However, sharing arbitrary entangled states is a strong requirement in practise. 
The key question for determining security of QPV is: how much entanglement do the attackers have to share to spoof the secure location?
Given that the honest prover implements some local operation on $n$ qubits, it is known that the upper bound on the minimal amount of entanglement required is $O(\exp(n))$ \cite{Beigi2011,Dolev}.
This would indicate that the protocol is reasonably secure.
However, the known lower bounds are only $O(n)$, \cite{Tomamichel2013}.
Therefore tightening the bounds on the optimal entanglement required to spoof QPV is an interesting open question.

\begin{defn}[Upper bounded linear attacks]\label{def: attacks}
An upper bounded linear attack on a QPV protocol is any attack where the pre-shared state held by the adversaries has mutual information upper bounded by $C_1n$ for some constant $C_1>1$, where $n$ is the number of qubits the unitary acts on in an honest implementation of the protocol.
\end{defn}

A protocol is secure against upper bounded linear attacks then there is no attack  for all constants $C_1$ that can break the security.

\subsubsection{Connection between QPV and holography}\label{sect intro asymp task}

\cite{May2019} introduced the concept of asymptotic quantum tasks, where a quantum information processing task is carried out in the bulk of AdS with inputs given and received at the conformal boundary of the spacetime.
In this context the holographic principle states that asymptotic quantum tasks are possible in the boundary if and only if they are possible in the bulk. 
Given AdS/CFT in (2+1)-dimensional bulk, it is possible to place the inputs and outputs of the task such that in the bulk there exists a region which is in the forward light cone of both inputs and the backward light cone of both outputs, whereas no such region exists on the boundary (see \cref{fig:pbqc}).
In the bulk, the task can be trivially completed: the inputs can be brought together and a local unitary applied before sending the two halves of the outputs to the corresponding delivery points.
However, while the holographic principle states that this implies the same task must be feasible on the boundary, there is no local region on the boundary where the computation can take place.
The boundary task must be completed via a non-local protocol making use of the entanglement in the boundary theory.

\cite{May2019} and \cite{May} suggested that such a holographic scattering experiment could be a interesting candidate for an attack on QPV with low entanglement. 

\begin{figure}[h!]
\centering
\includegraphics[trim={0cm 0cm 0cm 0cm},clip,scale=0.5]{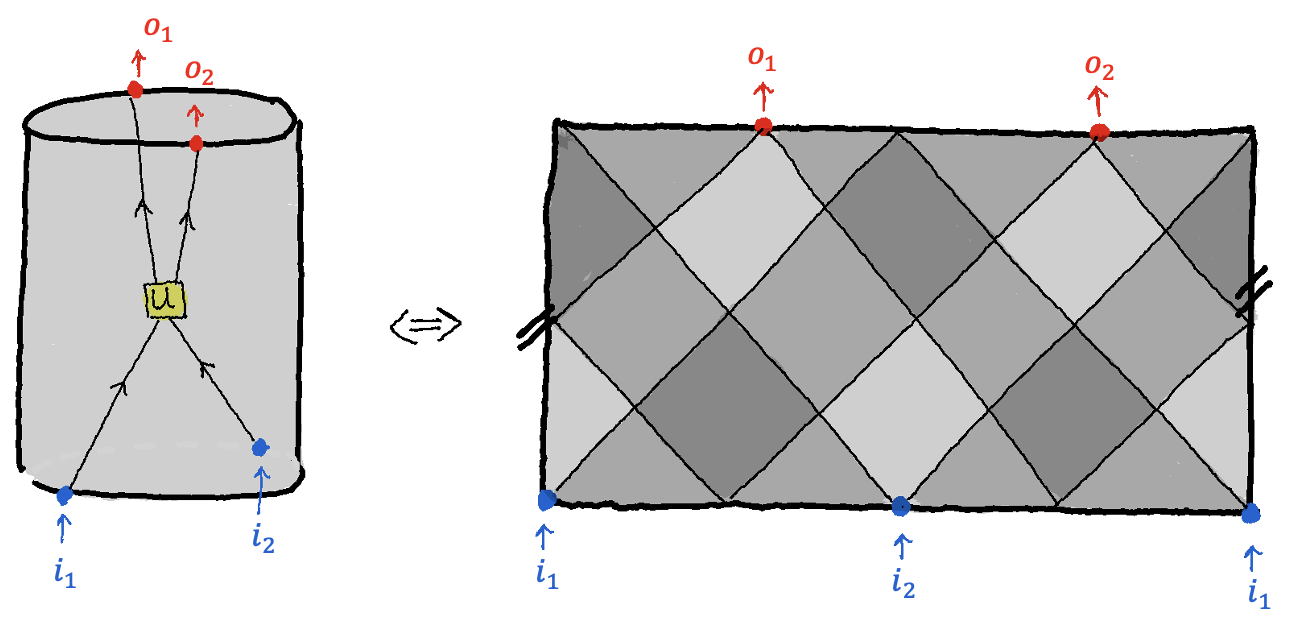}
\caption{The set up for our holography inspired proof of \cref{them PBQC sim}. On the left is the bulk theory where the light cones of $i_1$, $i_2$, $o_1$ and $o_2$ do intersect so a local computation is done in the centre of the bulk. This is parallel to the left side of \cref{fig:circuitdiag}. Due to the duality, the same result is obtained in the bulk and boundary theory. On the right in the boundary there is no region intersection of all light cones simultaneously and therefore the boundary theory achieves $U$ via non-local computation. We will later establish adversaries in the boundary and rounds that link this to the attack model in \cref{fig:circuitdiag}. }
\label{fig:pbqc}
\end{figure}

\subsection{Hamiltonian simulation}\label{sec:ham sim}

The next section will prove the main result of this paper -- a link between Hamiltonian simulation and quantum position-verification.
Before delving into that result, we will first introduce the basics of Hamiltonian simulation. 
In Hamiltonian simulation the goal is to engineer a simulator Hamiltonian $\Hsim$ to reproduce some aspect of a target Hamiltonian $\Htarget$ up to a controllable error.
There are a number of definitions for Hamiltonian simulation in the literature.
In order to compare the lower bounds on simulation  with the best known simulation techniques we need to pick a particular definition of simulation.
Here we choose to study the definition from \cite{Cubitt2019}, as this is the definition of simulation where errors are best understood.

Any simulation of $\Htarget$ by $\Hsim$ must involve encoding $\Htarget$ in $\Hsim$ in some way.
In \cite{Cubitt2019} it was shown that any encoding map $\mathcal{E}(\cdot)$ which satisfies basic requirements to be a good simulation must be of the form:
\begin{equation}
\mathcal{E}(A) = V \left(A \otimes P + \overline{A}\otimes Q \right)V^\dagger 
\end{equation}
where $V$ is an isometry, $P$, $Q$ are orthogonal projectors onto some ancillary systems and the overline denotes complex conjugation.

We will be interested in simulations of spin-lattice models, described by local-Hamiltonians where the target system can be decomposed as $\mathcal{H}_\textrm{target} = \otimes_{i=1}^n \mathcal{H}_i$.
We say an encoding is \textit{local} if the simulator system can be decomposed as $\mathcal{H}_\textrm{sim} = \otimes_{i=1}^n \mathcal{H}_i'$ such that $\mathcal{H}_i'$ corresponds to $\mathcal{H}_i$ operationally. 
Encodings capture the idea of one Hamiltonian perfectly replicating the physics of another.
To allow for some controllable error we introduce the concept of an approximate simulation:

\begin{defn}[Approximate simulation \cite{Cubitt2019}]\label{sim def}
We say that $\Hsim$ is a $(\Delta, \eta, \epsilon)$-simulation of $\Htarget$ if there exists a local encoding $\mathcal{E}(M) = V \left(M\otimes P + \overline{M} \otimes Q \right)V^\dagger$ such that:
\begin{enumerate}[i.]
\item There exists an encoding $\tilde{\mathcal{E}}(M) = \tilde{V}\left(M \otimes P + \overline{M}\otimes Q \right)\tilde{V}^\dagger$ such that $\tilde{V}\identity\tilde{V}^\dagger$ is the projector into the low ($<\Delta$) energy subspace of $\Hsim$  and $\norm{\tilde{V}-V}_\infty\leq \eta$
\item $\norm{(\Hsim)_{\leq \Delta}-\tilde{\mathcal{E}}(\Htarget)}_\infty\leq \epsilon$
\end{enumerate}
We say that a family $\mathcal{F}'$ of Hamiltonians can simulate a family $\mathcal{F}$ of Hamiltonians if, for any $\Htarget \in \mathcal{F}$ and any $\eta, \epsilon > 0$ and $\Delta > \norm{\Htarget}$ there exists $\Hsim \in \mathcal{F}'$ such that $\Hsim$ is a $(\Delta, \epsilon, \eta)$-simulation of $\Htarget$.
\end{defn}

This definition has three parameters describing the approximation in the simulation.
The physics of the target system is encoded in the subspace of $\Hsim$ with energy less than $\Delta$.
The high energy subspace of the simulator Hamiltonian contributes inaccuracies to the physics observed; in good simulations $\Delta$ is taken to be large to minimise these effects.
$\eta$ describes the error in the eigenstates since the local encoding describing the simulation does not map perfectly into the low energy subspace but instead is close to a general encoding that does.
Finally $\epsilon$ describes the error in the eigenspectrum.
A simple consequence of (ii) is that the low energy spectrum of the simulator is $\epsilon$-close to the spectrum of the target.
Given small $\epsilon$ and $\eta$, the eigenspectrum and corresponding states of the target are well approximated by the simulator.
In \cite{Cubitt2019} it is shown that approximate simulations preserve important physical quantities up to controllable errors:

\begin{lemma}[{\cite[Lem.~27, Prop.~28, Prop.~29]{Cubitt2019}}] \label{physical-properties}
Let $H$ act on $(\field{C}^d)^{\otimes n}$.
  Let $H'$ act on $(\field{C}^{d'})^{\otimes m}$, such that $H'$ is a $(\Delta, \eta, \epsilon)$-simulation of $H$ with corresponding local encoding $\mathcal{E}(M) = V(M \otimes P + \overline{M} \otimes Q)V^\dagger$.
  Let $p = \rank(P)$ and $q = \rank(Q)$.
  Then the following holds true.
  \begin{enumerate}[i.]
  \item Denoting with $\lambda_i(H)$ (resp.\ $\lambda_i(H')$) the $i$\textsuperscript{th}-smallest eigenvalue of $H$ (resp.\ $H'$), then for all $1 \leq i \leq d^n$, and all $(i-1)(p+q) \leq j \leq i (p+q)$, $|\lambda_i(H) - \lambda_j(H')| \leq \epsilon$.
  \item The relative error in the partition function evaluated at $\beta$ satisfies
    \begin{equation}
      \frac{|\mathcal{Z}_{H'}(\beta) - (p+q)\mathcal{Z}_H(\beta) |}{(p+q)\mathcal{Z}_H(\beta)} \leq \frac{(d')^m \ee^{-\beta \Delta}}{(p+q)d^n \ee^{-\beta \|H\|}} + (\ee^{\epsilon \beta} - 1).
    \end{equation}
  \item For any density matrix $\rho'$ in the encoded subspace for which $\mathcal{E}(\identity)\rho' = \rho'$, we have
    \begin{equation}
      \|\ee^{-\ii H't}\rho'\ee^{\ii H't} - \ee^{-\ii \mathcal{E}(H)t}\rho'\ee^{\ii \mathcal{E}(H)t}\|_1 \leq 2\epsilon t + 4\eta.
    \end{equation}
  \end{enumerate}
\end{lemma}

Families of Hamiltonians which are able to simulate arbitrary Hamiltonians are called \textit{universal}:

\begin{defn}[{Universal Hamiltonians~\cite[Def.~26]{Cubitt2019}}]
We say that a family of Hamiltonians is a universal simulator---or simply is universal---if any (finite-dimensional) Hamiltonian can be simulated by a Hamiltonian from the family.
\end{defn}

As we will see in the next section, a very good universal simulator for breaking QPV could simulate arbitrary target Hamiltonians for long times, with small errors, without blowing up the norm of the Hamiltonian.\footnote{These requirements would also make a very good simulator for practical applications, with the additional requirement that there it is important that the number of spins in the simulator system doesn't increase too much.}
We can formalise this as:
\begin{defn}[Very good universal simulators]\label{def: very good simulators}
We say that a universal family of $O(1)$-local Hamiltonians $\mathcal{M}$ is a very good universal simulator if for any target Hamiltonian $\Htarget$ acting on $n$ spins with locality $\mathcal{L}$ there exists a simulator Hamiltonian $\Hsim = \sum_x h_x \in \mathcal{M}$ which is a $(\Delta,\epsilon,\eta)$-simulation of $\Htarget$ where:
\begin{equation}
\max_i{\norm{h_i}} = \poly\left( \frac{\mathcal{L}^a}{\epsilon^b} \norm{\Htarget}, \frac{\mathcal{L}^x}{\epsilon^y \eta^z}\norm{\Htarget}\right)
\end{equation}
where $a+2b \leq 1$ and $x+2y + z\leq 1.$
\end{defn}

We can also consider loosening the above definition, to the case where we are only interested in target Hamiltonians that belong to some particular family of Hamiltonians:

\begin{defn}[Very good $\mathcal{M}_\textrm{target}$ simulators]\label{def: very good simulators 2}
We say that a universal family of $O(1)$-local Hamiltonians $\mathcal{M}$ is a very good $\mathcal{M}_\textrm{target}$ simulator if for a target Hamiltonian $\Htarget \in \mathcal{M}_\textrm{target}$ acting on $n$ spins with locality $\mathcal{L}$ and norm $\norm{\Htarget} =O\left( \frac{1}{n}\right)$ there exists a simulator Hamiltonian $\Hsim = \sum_x h_x \in \mathcal{M}$ which is a $(\Delta,\epsilon,\eta)$-simulation of $\Htarget$ where:
\begin{equation}
\max_i{\norm{h_i}} = \poly\left( \frac{\mathcal{L}^a}{\epsilon^b} \norm{\Htarget}, \frac{\mathcal{L}^x}{\epsilon^y \eta^z}\norm{\Htarget}\right)
\end{equation}
where $a+2b \leq 1$ and $x+2y + z\leq 1$.
\end{defn}

\subsection{QPV attack from causal toy models}\label{sect main result}

To construct an attack on QPV we will set up an asymptotic quantum task in the causal toy model with the features described in \cref{sect intro asymp task}.
The task is defined by a number of input $\{i_x\}$ and output $\{o_x\}$ points on the boundary of the toy model, a time $T_\textup{task}$ and a unitary $U$.
The task is completed successfully if quantum systems that are input at positions $\{i_x\}$ at time $t=0$ are acted on by the unitary $U$ and received at positions $\{o_x\}$ at time $T_\textup{task}$.

We will design the configuration in such a way that the unitary $U$ can be applied locally in the bulk, but cannot be applied locally in the boundary.
The set up is shown in \cref{fg time limit} for (2+1)D and (3+1)D systems.
The input and output points in both cases are on extreme points of the boundary.
This ensures that in the bulk the fastest way to carry out the task is to bring the input quantum systems to the centre of the bulk, apply the unitary, and send them to the output points. 
The time taken for a signal to travel from the boundary into the centre of the bulk and out to a (possibly different) boundary point is given by $Cn\tau^R$ for some constant $C$.
Applying the bulk unitary will take some non-zero time so we require $T_\textup{task} > Cn\tau^R$ in order to apply the unitary locally in the bulk.
For the configuration shown in \cref{fg time limit} it takes time $\frac{5}{4}Cn\tau^R$ to bring the input quantum systems together and send them to the output quantum systems if we restrict our protocol to only using the boundary.
Therefore, to ensure that there is no local boundary protocol for completing the task we require that $T_\textup{task} < \frac{5}{4}Cn\tau^R$ 
This means, we can allow bulk computations which run for time $t_U < \frac{1}{4}Cn\tau^R$ so that there is a local bulk procedure but only non-local boundary procedures.

\begin{figure}[h!]
\centering
\includegraphics[trim={0cm 0cm 0cm 0cm},clip,scale=0.45]{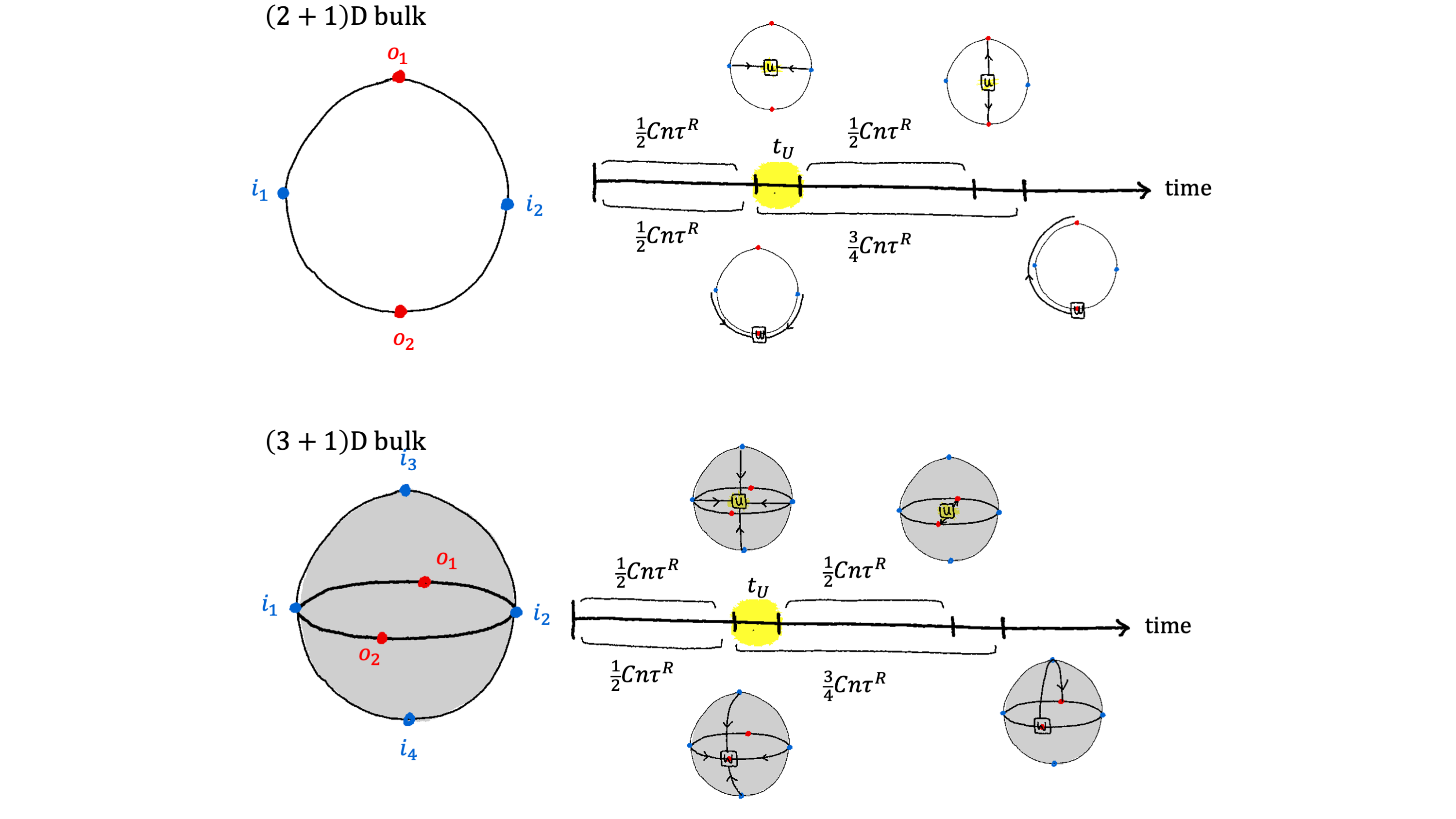}
\caption{The time for the local bulk computation is limited to $\frac{1}{4}Cn\tau^R$ otherwise a local spacetime region exists on the boundary. Top: a spacelike slice of (2+1)D bulk with inputs and outputs positioned on extreme points in the boundary. The minimum time taken for a particle to travel through the bulk is $Cn\tau^R$. Whereas the minimum time for the inputs to come together on the boundary and then travel to the outputs is $\frac{5}{4}Cn\tau^R$. Therefore, in the bulk a local unitary can be applied on the inputs for time $t_U\leq \frac{1}{4}Cn\tau^R$ while even an instantaneous local boundary protocol is causally impossible. Bottom: the same argument depicted with a (3+1)D bulk.}
\label{fg time limit}
\end{figure}

By the end of this section we will have proven the following result:
\begin{reptheorem}{them PBQC sim}[QPV and simulation]
If QPV is secure against upper bounded linear attacks (\cref{def: attacks}) then very good universal simulators (\cref{def: very good simulators}) cannot exist.
\end{reptheorem}

\subsubsection{Entanglement in the boundary} \label{sec:entanglement}

We will first demonstrate that the implementation of the $n$-qubit unitary through the non-local boundary protocol relies on a limited entanglement resource, specifically the boundary state.

\begin{figure}[h!]
\centering
\includegraphics[trim={0cm 0cm 0cm 0cm},clip,scale=0.35]{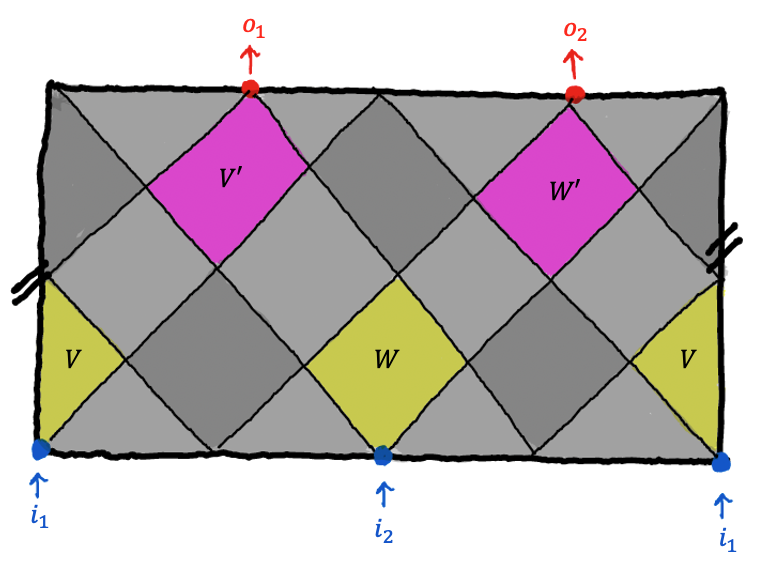}
\caption{Entanglement regions of interest on the boundary. The boundary regions are defined as follows: $V$ is the intersection of the forward lightcone of $i_1$ and the backwards lightcones of both outputs; $W$ is the intersection of the forward lightcone of $i_2$ and the backwards lightcones of both outputs; $V$ is the intersection of the backward lightcone of $o_1$ and the forward lightcones of both inputs; $W'$ is the intersection of the backward lightcone of $i_1$ and the forward lightcones of both inputs.}
\label{fg: boundary entanglement}
\end{figure}

As illustrated in~\cref{fg: boundary entanglement}, due to causality constraints the boundary protocol can be restricted to two rounds of operations. 
The boundary regions $V$, $W$, $V'$ and $W'$ are defined in~\cref{fg: boundary entanglement} by the inputs/outputs they are causally connected to.
The first round of operations take place in $V$ and $W$ separately and a second round of operations in $V'$ and $W'$ with communication allowed in between. 
Therefore, the amount of entanglement used in the boundary protocol can be captured by the mutual information between the two pairs of boundary regions $I(V:W)$ and $I(V':W')$.
By virtue of time reversal symmetry in the specific geometry of the inputs and outputs considered here these quantities are equal.
We will show that in the protocol where a local unitary on $n$ qubits is performed in the bulk, that the mutual information between these regions is upper bounded by
\begin{equation}\label{eqn upper bound}
I(V:W) \leq \frac{1}{2}C_1 n,
\end{equation}
for some constant $C_1\geq 4\log(2)R $.

Mutual information is given by the difference in von-Neumann entropies,
\begin{equation}\label{eqn: mutual info}
I(V:W) = S(\rho_V) + S(\rho_W) - S(\rho_{V\cup W}).
\end{equation}
Hence in order to upper bound the mutual information we need to both upper and lower bound the entropy of a boundary subregion. 
\cite{Happy} section 4 details entropy bounds for perfect tensor networks.

First we rephrase a general tensor network result to make use of subsequently:
\begin{lemma}[General tensor network state entanglement upper bound]\label{lm: entanglement upper bound}
Given a tensor network state, $\ket{\psi}$ in $\mathcal{H}$ and a subspace associated with some tensor legs $\mathcal{H}_A \subset \mathcal{H}$.
The entanglement entropy is upper bounded by
\begin{equation}
S(\rho_A) \leq \min_{\gamma_A}  \log d_i ,
\end{equation}
where $\gamma_A$ is a cut through the tensor network separating $\mathcal{H}_A$ and $\mathcal{H}/\mathcal{H}_A$ that defines a decomposition of the tensor network into a contract of two tensors $P$ and $Q$ such that,
\begin{equation}
\ket{\psi} = \sum_{a,b,i} P_{ai}Q_{ib} \ket{a}\ket{b},
\end{equation}
with $\ket{a}$, $\ket{b}$ an orthogonal basis of $\mathcal{H}_A$ and $\mathcal{H}/\mathcal{H}_A$ respectively.
$d_i$ is the Hilbert space dimension associated with the contraction between $P$ and $Q$.
\end{lemma}

\begin{proof}
The reduced density matrix on subsystem $\mathcal{H}_A$ is given by,
\begin{align}
\rho_A &= \trace_{\mathcal{H}/\mathcal{H}_A} \left[\ket{\psi}\bra{\psi} \right] \\
& = \trace_{\mathcal{H}/\mathcal{H}_A} \left[ \sum_{a,b,i,a',b',i'} P_{ai}Q_{ib}(Q_{i'b'})^\dagger (P_{a'i'})^\dagger \ket{a}\bra{a'} \otimes \ket{b}\bra{b'}\right]\\
& = \sum_{a,b,i,a',b',i'} P_{ai}Q_{ib}\overline{Q}_{b'i'}\overline{P}_{i'a'} \ket{a}\bra{a'} \otimes \sum_w \braket{w|b}\braket{b'|w}\\
& =  \sum_{a,b,i,a',i'} P_{ai}Q_{ib}\overline{Q}_{bi'}\overline{P}_{i'a'} \ket{a}\bra{a'}\label{eqn rank rhoA},
\end{align}
where $\ket{a}$ is an orthonormal basis for $\mathcal{H}_A$, $\ket{b}$ and $\ket{w}$ are orthonormal basis for $\mathcal{H}/\mathcal{H}_A$ and overline denotes taking the complex conjugate. 
The Schmidt rank of $\rho_A$ is the number of strictly positive coefficients, $\alpha_j$, in the Schmidt decomposition, $\rho_A = \sum_j \alpha_j \ket{a}\bra{a}$.
From \cref{eqn rank rhoA} one can see that the rank of $\rho_A$ is upper bounded by the number of terms in the sum over $i$,
\begin{equation}
\textup{Rank}[\rho_A] \leq d_i.
\end{equation}
The von Neumann entropy is upper bounded by log of the rank, hence for all choices of $\gamma_A$,
\begin{equation}
S(\rho_A) \leq d_i.
\end{equation}
To obtain the tightest upper bound minimise over the different choices of $\gamma_A$.
\end{proof}

Next we discuss entropy lower bounds.
A lower bound was given in \cite{Happy} for product bulk states or holographic states.
However, since an entangled bulk state can only generate additional contributions to entanglement, the same lower bound is valid. 
The result is given in terms of greedy geodesics.
 
\begin{defn}[Greedy algorithm: see \cite{Happy}]\label{defn: greedy}
Start with a cut through a tensor network. 
One step of the algorithm consists of identifying a tensor for which the current cut crosses at least half of its legs. 
The update then moves the cut to include this tensor.
This procedure is iterated until arriving at the local minimum cut where adding or removing any single tensor does not reduce the cut length. 
\end{defn}

Figure 7 of \cite{Happy} demonstrates that the greedy algorithm can fail to find matching minimal geodesics starting from a boundary region $\gamma^\star_X$ and the complementary boundary region $\gamma^\star_{X^c}$. 
For example, in a model with negative curvature, the algorithm leaves residual bulk regions between $\gamma^\star_X$ and $\gamma^\star_{X^c}$ for disconnected regions or connected regions with free bulk indices. 

\begin{lemma}[Boundary entropy lower bound: see Theorem 3 of \cite{Happy}]\label{lm: entanglement lower bound}
Given a holographic quantum error correcting code contracted from perfect tensors with bond dimension $\nu$.
For a holographic state or code with arbitrary bulk state contracted, if $A$ is a (not-necessarily-connected) boundary region and $A^c$ is its complement, then the entropy of the reduced boundary state on $A$ satisfies,
\begin{equation}
\log(\nu)\abs{\gamma_A^\star \cap \gamma_{A^c}^\star} \leq S(\rho_A),
\end{equation}
where $\gamma_A^\star$ is the greedy geodesic obtained by applying the greedy algorithm to $A$ and $\gamma_{A^c}^\star$ is the greedy geodesic obtained by applying the greedy algorithm to $A^c$.
\end{lemma}
\begin{proof}
See \cite{Happy} section 4.2.
\end{proof}

By \cref{lm: entanglement upper bound} and \cref{lm: entanglement lower bound} the entropy of a boundary subregion $X$ with the bulk state with entanglement $S_\textup{bulk}$ is upper and lower bounded by,
\begin{equation}\label{eqn: entropy bounds}
\log(\nu) \abs{\gamma_X^\star \cap \gamma_{X^c}^\star} \leq S(\rho_X) \leq \log(\nu) \abs{\gamma_X} .
\end{equation}
Substituting appropriately from \cref{eqn: entropy bounds} into \cref{eqn: mutual info} gives the following upper bound on the mutual information,
\begin{equation}\label{eqn: upper 1}
I(V:W) \leq \log(\nu) \left[\abs{\gamma_V} + \abs{\gamma_W} - \abs{\gamma_{V\cup W}^\star \cap \gamma_{(V\cup W)^c}^\star } \right] .
\end{equation}

\begin{figure}[h!]
\centering
\includegraphics[trim={0cm 0cm 0cm 0cm},clip,scale=0.4]{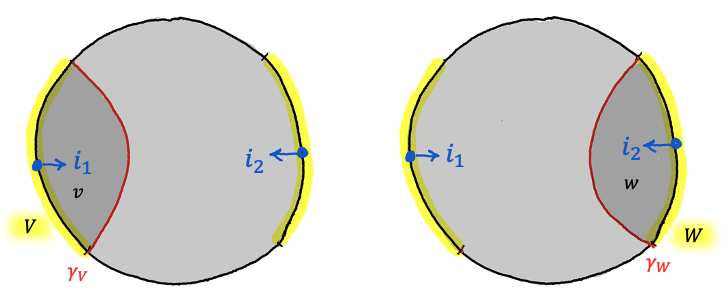}
\caption{Bulk subregions containing only one input/output have bulk entropy upper bounded by $\floor{n/2}\log(2)$. An equivalent picture exists for $V'$ and $W'$. }
\label{fg: vwbulk}
\end{figure}

We can use information about our protocol to quantify these contributions to the bounds. 
The quantities $\abs{\gamma_V} ,\abs{\gamma_W}, \abs{\gamma_{V\cup W}^\star \cap \gamma_{(V\cup W)^c}^\star } $ appearing in \cref{eqn: upper 1} are all lengths of cuts through the tensor network. 
In our network each "tensor leg" actually consists of a bundle of $n$ qubits, hence all lengths are clearly proportional to $n$.  

In the case where the greedy algorithm fails maximally $ 0\leq \abs{\gamma_{V\cup W}^\star \cap \gamma_{(V\cup W)^c}^\star } $, in this case the lower bound in \cref{lm: entanglement lower bound} simply corresponds to entropy being a non-negative function. 
Hence, 
\begin{equation}
I(V:W) \leq \log(2)\left( \abs{\gamma_V} + \abs{\gamma_W} \right).
\end{equation}
Define $\abs{\Gamma_X} := \frac{1}{n}\abs{\gamma_X}$ as the number of index bundles crossed by cut $\gamma$, so that, 
\begin{equation}
I(V:W) \leq \log(2)\left( \abs{\Gamma_V} + \abs{\Gamma_W}  \right)n.
\end{equation}
Therefore:
\begin{equation}
I(V:W) \leq \frac{1}{2}C_1 n,
\end{equation}
for some constant independent of $n$, $C_1 = 2\log(2)\left( \abs{\Gamma_V} + \abs{\Gamma_W} \right)$, and equivalently for $V'$ and $W'$.
The cuts through the network can be loosely upper bounded by $\abs{\Gamma_V},\abs{\Gamma_W} \leq 2R$ - since this cut goes to the centre of the bulk and returns to the boundary, therefore
\begin{equation}
I(V:W) \leq \log(2)2Rn.
\end{equation}
The total entanglement for the non-local protocol is then $I(V:W)  + I(V':W') \leq C_1n $ with $C_1\geq 4\log(2)R>1$.
Note that since the two rounds concern different sets of boundary legs the total amount of entanglement in the pre-shared boundary state between the two adversarial parties also includes the mutual information in the second round.  

\subsubsection{Proof of main theorem}

We are now in a position to prove the main result. 

\begin{theorem}[QPV and simulation] \label{them PBQC sim}
If QPV is secure against upper bounded linear attacks (\cref{def: attacks}) then very good simulators (\cref{def: very good simulators}) cannot exist.
\end{theorem}

\begin{proof}
Let $\mathcal{A}$ be the statement that quantum position verification is secure against upper bounded linear attacks.
Let $\mathcal{B}$ be the statement that very good universal simulators exist.
The above theorem is then summarised by $\mathcal{A}\implies \neg\mathcal{B}$.
Proof by contrapositive, we will show that if very good universal simulators do exist ($\mathcal{B}$) then we construct a non-local attack on the QPV protocol implementing a general unitary on $n$ qubits using at most $C_1n$ entanglement ($\neg \mathcal{A}$).

Consider a tensor network toy model of AdS/CFT with a (2+1)D bulk and a (1+1)D boundary as outlined in the previous section.
The tensors are chosen to be perfect tensors and each tensor "leg" is a bundle of $n$ qubits.
Set up two input points, $i_1,i_2$, at opposite points on the boundary, and two output points $o_1,o_2$ at spatial positions rotated by $\frac{\pi}{2}$ (see \cref{fig:pbqc}).

The honest protocol in the language of QPV will take place in the bulk. 
We initialise two quantum states on the bulk degrees of freedom directly adjacent to the  $i_1,i_2$ locations on the boundary.
Each input state is on $\frac{n}{2}$ qubits so can be initialised on half the bundle of indices that make up a single bulk tensor leg.
These input states are then sent into the centre of the bulk by applying a Hamiltonian that generates the SWAP operator between the bulk degree of freedom carrying the input state, and an adjacent bulk tensor leg which is one step closer to the centre of the tensor network (depending on the choice of tensor network there may not be a unique choice of adjacent bulk tensor leg - this is unimportant for us, any choice will lead to an equivalent protocol).
Note that this Hamiltonian is time dependent -- the interaction between tensor legs at depth $r$ and depth $r-1$ are only turned on when the input state reaches the tensor at depth $r$.
Each SWAP operation takes time $O(n\tau^{R-r})$ where $r$ is the distance from the centre of the tensor network and $\tau$ is a constant associated with the choice of hyperbolic honeycombing in the tensor network.
Once the two input states reach the central bulk tensor we apply a Hamiltonian that generates the desired local unitary $U$.
To ensure that the boundary dynamics for this evolution gives a protocol for non-local quantum computation we require that the Hamiltonian generating the bulk unitary is applied for time at most $t_U < \frac{C}{4}n\tau^{R}$ to maintain the impossibility of a local unitary in the boundary (see \cref{fg time limit}).
We then send the output states to the bulk tensor legs directly adjacent to the $o_1,o_2$ locations on the boundary.
As before, the process of sending the states to the boundary is achieved by applying a time dependent Hamiltonian which generates the SWAP operator between the bulk degree of freedom carrying the output state, and an adjacent bulk tensor leg which is one step closer to the output point (again there may be more than one choice of shortest path from the bulk central tensor to the output point - any choice of shortest path will do).

Before we present the dishonest protocol in terms of QPV language we first examine the boundary dynamics that are dual to the bulk process outlined above. 
Since the tensor network is an error correcting code, when examining the bulk protocol on the boundary there is a choice about which portions of the boundary are used to reconstruct bulk operators.
When the inputs are being sent into the centre of the bulk we adopt the `natural' choice of pushing directly out towards the boundary -- this means we are pushing out to a non-contiguous boundary region consisting of two portions, one centred on $i_1$ and the other centred on $i_2$
When the inputs reach the centre of the bulk, for the first half of the time that $U$ is being applied we continue to push out towards the same non-contiguous boundary region that was being used for the final SWAP.
For the second half of the implementation of $U$ we push out to a different non-contiguous boundary region.
It is again made up of two portions, this time one centred on $o_1$, the other on $o_2$. 
When the signal is being sent out from the bulk to the output points we again adopt the `natural' choice of pushing directly out towards the boundary towards regions centred on $o_1$ and $o_2$.

When we simply push the bulk operators out to the boundary using the tensor network we find that the resulting boundary operators are highly non-local.
To construct local boundary operators we will apply a very good simulation (\cref{def: very good simulators}, which exists by assumption) to the non-local boundary dynamics outlined above. 
Recall from \cref{physical-properties} that the error in a time evolved state after a simulation is upper bounded by $2\epsilon t + \eta$.
We apply a number of simulations on the boundary -- we simulate the boundary dynamics dual to each SWAP operation in the bulk and the boundary dynamics dual to the unitary $U$ in the bulk. 
Since these boundary operations happen sequentially, the error from one simulation feeds into the next, and the total error in the final boundary state is given by:
\begin{equation}
E = \sum_{r=0}^R (2\epsilon^{\textrm{in}}_r t_r + \eta^{\textrm{in}}_r)  + 2\epsilon_U t_U + \eta_U + \sum_{r=0}^R (2\epsilon^{\textrm{out}}_r t_r + \eta^{\textrm{out}}_r) 
\end{equation}
where $t_r = O(n\tau^{R-r})$ is the time taken for the SWAP operation between tensor bulk legs at $r$ and $r-1$, $t_U < \frac{1}{4}n\tau^R$ is the time the central bulk unitary is applied for, $\epsilon_U, \eta_U$ are the errors in the simulation of the boundary dynamics dual to the central bulk unitary,  $\epsilon^{\textrm{in}}_r, \eta^{\textrm{in}}_r$ ($\epsilon^{\textrm{out}}_r, \eta^{\textrm{out}}_r$) denote the errors in the simulation of the boundary dynamics dual to the SWAP operation beiing applied between tensor bulk legs at $r$ and $r-1$ on the way in (on the way out).
We require that this total error is much less than one, we can achieve this by simulations that achieve error scalings given by:
\begin{equation}
\epsilon^{\textrm{in}}_r = \epsilon^{\textrm{out}}_r = \epsilon_r = \frac{\tau^{-2(r-R)}}{n^2} \textrm{\ \ \ , \ \ \ } \eta^{\textrm{in}}_r = \eta^{\textrm{out}}_r = \eta_r = \frac{\tau^{-{(R-r)}}}{n}  \textrm{\ \ \ , \ \ \ }  \epsilon_U = \epsilon_R \textrm{\ \ \ , \ \ \ }  \eta_U = \eta_R
\end{equation}
This gives a total error scaling as:
\begin{equation}
E = O\left(\frac{1}{n}\right)
\end{equation}
which can be made arbitrarily small by increasing $n$.
By assumption we can achieve a boundary interaction dual to the bulk interaction at tensor layer $r$, with the boundary interaction strength scaling as:
\begin{equation}
\norm{h_r} = \poly\left(\frac{\mathcal{L}^a}{\epsilon^b}\norm{H_r},\frac{\mathcal{L}^x}{\epsilon^y\eta^z}\norm{H_r} \right)
\end{equation}
where $a+2b \leq 1$ and $x+2y+z \leq 1$ and $\norm{H_r}$ is the interaction strength of the dual bulk interaction at tensor layer $r$.
Therefore we have that:
\begin{equation}
\norm{h_r} = \poly\left(\frac{(n\tau^{R-r})^a}{\left(\frac{\tau^{r-R}}{n}\right)^{2b}}\frac{\tau^{r-R}}{n},\frac{(n\tau^{R-r})^x}{\left(\frac{\tau^{r-R}}{n}\right)^{2y}\left(\frac{\tau^{r-R}}{n}\right)^z}\frac{\tau^{r-R}}{n} \right) = O(1).
\end{equation}
These error scalings achieve a causal boundary simulation, where the speed of information propagation in the boundary is shown in \cref{fg: boundary entanglement}.

Now we turn to look at the dishonest protocol in the language of QPV.
Even with the causal simulation outlined above the boundary protocol acts over the entire boundary.
However, the entire protocol can be very well approximated by adversaries holding only the tensor legs corresponding to regions $V,V',W,W'$ in \cref{fg: boundary entanglement} since the Lieb-Robinson speed of the boundary ensures that to good approximation these are the only regions which are causally connected to one input and both outputs, or vice versa.\footnote{This means that if we trace out the entire boundary except the regions $V,V',W,W'$ the final boundary state will be approximately the same as if we acted on the entire boundary.}
Therefore at the beginning of the protocol we will let our first attacker, Alice, hold the boundary state corresponding to the tensor legs $V \cup V' \setminus W$ and our second attacker hold the boundary state corresponding to the tensor legs $W \cup W' \setminus V$.
Due to constants in the $O(1)$ Lieb-Robinson velocity and finite size effects in the tenor network, the region $V \cup V' $ may have some overlap with $W \cup W' $.
In the first round of the attack Alice (Bob) applies the Hamiltonian terms resulting from the very good simulation of the boundary dynamics to the qubits in $V$ ($W$) up to and including the first half of the boundary dynamics dual to the central bulk unitary.
Alice then sends Bob any qubits in her possession which are in the region $W'$, and Bob sends Alice any qubits in his possession which are in the region $V'$ (these qubits would result from any overlap between the regions $V$ with $W'$ and $W$ with $V'$).\footnote{Note that in fact this communication step will overlap with the computation steps since any qubits which need to be sent are at the edges of the relevant regions and do not need to be acted on for the whole of the first or second round.} 
Alice (Bob) then applies the Hamiltonian terms resulting from the very good simulation of the boundary dynamics to the qubits in $V'$ ($W'$) starting from the second half of the boundary dynamics dual to the central bulk unitary, and finishing with the simulation terms resulting from the boundary dynamics dual to the SWAP operations that send the information back to the boundary.

As outlined in \cref{sec:entanglement} the entanglement between these regions is upper-bounded by $C_1n$ for some constant $C_1\geq4\log(2)R $.
Therefore, if very good simulations exist, such that $\mathcal{B}$ is true, we can construct a  non-local attack on quantum position-verification with an error in the final state of $\epsilon t + \eta \ll 1$ using at most linear entanglement by using the HQECC and Hamiltonian simulation techniques to transform the local bulk protocol into a non-local boundary one that implements the same task. 
The amount of quantum communication is also upper bounded by a linear and there's no classical communication so all resources are at most linear.
\end{proof}

\cref{them PBQC sim} holds for general unitaries.
However, we can derive a corollary by restricting the local unitary in the position-based protocol to those generated by certain families of Hamiltonians.
This demonstrates that even when considering a limited family of unitaries the same logic bounds the interaction strength of the simulator Hamiltonian:

\begin{corollary}[QPV and simulation] \label{them PBQC sim entire physics}
If every QPV that implements a unitary $U = e^{i H_U t}$ for $H_U \in \mathcal{M}_U$ is secure against upper bounded linear attacks (\cref{def: attacks}) then very good $\mathcal{M}_\textrm{target}$ simulators (\cref{def: very good simulators 2}) cannot exist, where $\mathcal{M}_\textrm{U}$ and $\mathcal{M}_\textrm{target}$ are related by an isometry.
\end{corollary}
\begin{proof}
This follows immediately from the proof of \cref{them PBQC sim}, where instead of allowing the unitary implemented in the bulk to be general, we require that is generated by a restricted family of Hamiltonians (e.g. $k$-local Hamiltonians). The form of the $\Htarget$ comes from pushing the Hamiltonian that generates the unitary through the tensor network, as outlined in the proof of \cref{them PBQC sim}.
\end{proof}
See \cref{discussion} for more discussion on these results.

\subsection{Best known simulation techniques}\label{sect: best known}

To assess the importance of the lower bound on resources required for simulation from \cref{them PBQC sim} we can compare it to the current best known simulation techniques.

In Appendix \ref{app:non perturbative} we modify the universal Hamiltonian construction of \cite{Kohler2020} to one where the interactions themselves depend on the Hamiltonian to be simulated (in a precise manner).
In this way we are able to exponentially improve on the required interaction strength in the simulator system.
However, for general Hamiltonians this only takes us from a doubly exponential scaling of interaction strength with $N$ to a single exponential scaling.
This is clearly significantly worse than the optimum scaling outlined in \cref{them PBQC sim}.
However, when restricting to sparse Hamiltonians we  are able to construct a family of 2-local Hamiltonians $H = \sum_i h_i$ which can simulate an arbitrary sparse Hamiltonian $\Htarget$ with interaction strength satisfying $\max_i{\norm{h_i}} = \left(\poly\left( \frac{N^a}{\epsilon^b} \norm{\Htarget}, \frac{N^x}{\epsilon^y \eta^z}\norm{\Htarget}\right)\right)$ where $a+2b>2\frac{1}{3}$ and $x+2y+z > 2\frac{1}{2}$.\footnote{For target Hamiltonians that are $O(1)$-sparse we can achieve $a+2b=2\frac{1}{3}$ and $x+2y+z = 2\frac{1}{2}$. For more general $\poly(N)$ sparse Hamiltonians we can achieve $a+2b>2\frac{1}{3}$ and $x+2y+z > 2\frac{1}{2}$ where the exact values depend on the degree of the polynomial.} 

Note that when considering the result about sparse Hamiltonians in the setting of \cref{them PBQC sim entire physics} it is $\Htarget$ which we require to be sparse. 
A sufficient condition for $\Htarget$ to be sparse is that the Hamiltonian generating the unitary in the QPV protocol have Pauli rank $O(\poly(n))$ (where the Pauli rank is the number of terms in the Pauli decomposition of the Hamiltonian).
This set of Hamiltonians includes, but is not limited to, $k$-local Hamiltonians for constant $k$.

Although the achievable simulations for sparse Hamiltonians are close to the bounds in \cref{them PBQC sim}, they still do not allow us to construct interesting attacks on QPV.
This is because in order to keep the part of the error that scales with time small we have to focus on very low norm bulk Hamiltonians.
However, this means that the unitary we are carrying out is close to the identity, and its simulation is swamped by the part of the error that is time-independent. 

The simulation techniques from Appendix \ref{app:non perturbative} are non-perturbative.
In \cref{sect pert const} we investigate perturbative simulation techniques. 
However, there are no regimes where these give bounds as good as the non-perturbative alternatives.
This may be because we have not optimised the perturbative techniques for cases where we are simulating Hamiltonians that have large locality but are sparse.
Optimising these types of simulation has been investigated and it has been suggested that they generically require  $\max_i{\norm{h_i}} = O\left(\norm{\Htarget}\poly\left(\frac{1}{\epsilon^k}\right)\right)$ to simulate a $k$-local $\Htarget$ with a 2-local $H = \sum_i h_i$ while keeping the error and the number of spins fixed \cite{Cao2017,Bravyi2008,Cao2018}.
As in the case of the non-perturbative techniques, any attack on QPV using these perturbative techniques would only be applicable for the case of a unitary that was close to trivial.

\section{Discussion}\label{discussion}

\subsection{Toy models of AdS/CFT}

The rescaling of the Hamiltonian in tensor network toy models of AdS/CFT gives toy models with more of the expected causal features of the duality.
However, the toy models outlined here are still lacking key features such as Lorentz invariance, and conformal symmetry on the boundary.
Understanding whether it is possible to incorporate those other features of AdS/CFT into toy models could help aid our understanding of which features of holography are fundamentally gravitational.
Moreover, the time dilation expected in AdS/CFT does not appear naturally in these types of toy models, the way it does in e.g. \cite{masanes2023discrete}.
Constructing toy models where time dilation appears naturally but which also have bulk degrees of freedom (which \cite{masanes2023discrete} does not) remains an open question and would be the ideal tool for understanding QPV protocols in holography in a toy setting.

\subsection{The connection between Hamiltonian simulation and QPV}

To build on our main theorem there are two key areas for future research.
The first is attempting to construct simulators which achieve the bounds in \cref{them PBQC sim}.
While doing this for completely general Hamiltonians is a daunting task, it may be possible to do for restricted families of Hamiltonians.
Indeed, for certain families of sparse Hamiltonians we can already get close to the bounds, although not close enough to construct interesting attacks on QPV (Appendix \ref{app:non perturbative}).
As highlighted by \cref{them PBQC sim entire physics}, constructing simulators that could achieve the bounds from \cref{them PBQC sim} for restricted families of target Hamiltonians would already demonstrate that there exist certain families of unitaries that cannot be used as the basis for QPV protocols that are secure against linear entanglement.
For example, the families of sparse Hamiltonians we construct simulations for in Appendix \ref{app:non perturbative} correspond to unitaries which are generated by $k$-local Hamiltonians.
Therefore, if we could close the gap between the results of Appendix \ref{app:non perturbative} and the bounds in \cref{them PBQC sim} we would demonstrate that unitaries which are generated by $k$-local Hamiltonians cannot be used as the basis for QPV protocols that are secure against linear entanglement.
On the other hand, if unitaries generated by $k$-local Hamiltonians can be used as the basis for QPV protocols which are secure against linear entanglement this immediately implies that the results in Appendix \ref{app:non perturbative} are close to optimal for simulating sparse Hamiltonians. 


We note that a recent work \cite{harley2023going} shows that for \emph{modular} simulations of a $k$-local Hamiltonian (i.e. simulations where each $k$-local term in $\Htarget$ is simulated independently by a gadget) the interaction strength of the simulator Hamiltonian can be lower bounded by $O(n^{\frac{1}{2}})$.
The no go result in \cite{harley2023going} does not apply to simulations where the whole of $\Htarget$ is simulated at once (e.g. history state simulations, see Appendix \ref{app:non perturbative}) and hence it does not rule out the `very good' simulations in \cref{them PBQC sim}.
They also propose a new definition of dynamical simulation that is weaker than \cref{sim def} in that it only requires simulating subsets of observables, rather than the entire state.
Within this definition they are able to circumvent their no go result by utilising repeated measurement of ancilla spins to construct a dynamical simulation of a $O(1)$-local Hamiltonian with $O(1)$ interaction strengths.
This technique relies on a Lieb-Robinson type argument to bound the evolution of a very local observable under a Hamiltonian acting on a larger system.
The condition that the Hamiltonian is geometrically $O(1)$-local is necessary in their error bound and therefore cannot be applied to a general Hamiltonian acting on $n$ qubits -- as is the case for the boundary Hamiltonian in the toy model -- to obtain a QPV attack.

There are simpler lower bounds on the interaction strengths required to simulate $n$ local Hamiltonians with geometrically $2$ local terms due to causality arguments.
We explore these in Appendix \ref{appen: LR}.
The bound is linear in $n$ and therefore weaker than that of \cref{them PBQC sim} due to the exponents ($a+b>1$).
Furthermore, allowing an error in the state does not alter the scaling -- a key feature of this result. 
We leave it to future work to investigate whether causality requirements may produce new simulation lower bounds.

If there are states with no long-range entanglement in the simulator subspace then the conditional restrictions on `very good simulations' can be shown unconditionally via a  Lieb-Robinson argument.
However, this argument cannot be applied when there is long-range entanglement, as is the case for the position-based verification set up with $O(n)$ entanglement between opposite sides of the boundary. 
It could be possible to simulate a signalling Hamiltonian with a non-signalling Hamiltonian if the large norm difference between operators -- which is guaranteed to exist due to Lieb-Robinson bounds -- originates from outside the simulator subspace. 

\section*{Acknowledgements}
The authors would like to thank Alex May and Kfir Dolev for insightful discussions.
H.\,A. ~is supported by EPSRC DTP Grant Reference: EP/N509577/1 and EP/T517793/1.
P.\,H. is supported by AFOSR (award FA9550-19-1-0369), DOE (Q-NEXT), CIFAR and the Simons Foundation
T.\,K. is supported by the DOE QuantISED grant DE-SC0020360, SLAC (Q-NEXT) and QFARM.
D.\,P.\,G. and T.\,K. also acknowledges financial support from the Spanish Ministry of Science and Innovation (“Severo Ochoa Programme for Centres of Excellence in R\&D” CEX2019-000904-S and grant PID2020-113523GB-I00).
This work has been supported in part by the EPSRC Prosperity Partnership in Quantum Software for Simulation and Modelling (grant EP/S005021/1), and by the UK Hub in Quantum Computing and Simulation, part of the UK National Quantum Technologies Programme with funding from UKRI EPSRC (grant EP/T001062/1).
The work has been financially supported by Universidad Complutense de Madrid (grant FEI-EU-22-06), by the Ministry of Economic Affairs and Digital Transformation of the Spanish Government through the QUANTUM ENIA project call - QUANTUM SPAIN project, and by the European Union through the Recovery, Transformation and Resilience Plan - NextGenerationEU within the framework of the Digital Spain 2026 Agenda. 
This research was supported in part by Perimeter Institute for Theoretical Physics. Research at Perimeter Institute is supported by the Government of Canada through the Department of Innovation, Science, and Economic Development, and by the Province of Ontario through the Ministry of Colleges and Universities. 

\begin{appendices}
\addtocontents{toc}{\protect\setcounter{tocdepth}{1}}

\section{An explicit non-perturbative construction of an attack on QPV} \label{app:non perturbative}

In \cite{Kohler2020} a history-state method for Hamiltonian simulation was outlined.
This simulation technique leverages the ability to encode computation into ground states of local Hamiltonians via computational history states.
A computational history state $\ket{\Phi}_{CQ} \in \mathcal{H}_C \otimes \mathcal{H}_Q$ is a state of the form
\begin{equation}
\ket{\Phi}_{CQ} = \frac{1}{\sqrt{T}} \sum_{t=1}^{T} \ket{t}\ket{\psi_t},
\end{equation}
where $\{\ket{t}\}$ is an orthonormal basis for $\mathcal{H}_C$ and $\ket{\psi_t} = \Pi_{i=1}^tU_i\ket{\psi_0}$ for some initial state $\ket{\psi_0}\in \mathcal{H}_Q$ and set of unitaries $U_i \in \mathcal{B}(\mathcal{H}_Q)$.

$\mathcal{H}_C$ is the clock register and $\mathcal{H}_Q$ is  the computational register.
If $U_t$ is the unitary transformation corresponding to the $t$\textsuperscript{th} step of a quantum computation (in any model of computation, e.g. quantum circuits, quantum Turing machines, or more exotic models of computation), then $\ket{\psi_t}$ is the state of the computation after $t$ steps.
The history state $\ket{\Phi}_{CQ}$ then encodes the evolution of the quantum computation.
History state Hamiltonians are constructed so that their ground states are computational history states. 
In \cite{Kohler2020} a method is derived to construct universal Hamiltonians via history state methods.
Here we use a slight modification of that method\footnote{In \cite{Kohler2020} the aim was to minimise the number of parameters needed to describe a universal quantum simulators. Here we are instead interested in minimising the norm of the simulator interactions. For that reason, instead of encoding the target Hamiltonian in a feature of the simulator system, we instead encode it directly into the interactions.}.

We construct a history state Hamiltonian, $\Hpe$, which encodes the Turing machine model of computation.
The input to the Turing machine is a set of `physical spins' -- the state of these spins is left arbitrary; they encode the state of the system being simulated.
The computation being encoded is a long period of `idling' (i.e. doing nothing), followed by quantum phase estimation applied to the physical spins with respect to the unitary generated by $\Htarget$.
The output of the $\Hpe$ computation is the energy of the physical spins with respect to the Hamiltonian $\Htarget$.
This energy is written to the Turing machine tape in the form $E = \norm{\Htarget} (m_\alpha \sqrt{2}-m_\beta)$ by writing $m_\alpha$ $\alpha$ symbols and $m_\beta$ $\beta$ symbols to the tape.

The simulator Hamiltonian is then given by:
\begin{equation}\label{eq:first order sim}
\Hsim = J \Hpe + T \norm{\Htarget} \sum_{i=0}^{N-1}\left(\sqrt{2}\Pi_\alpha^{(i)} -\Pi_\beta^{(i)} \right)
\end{equation}
where $T$ is the total number of time steps in the computation encoded by $\Hpe$, $N$ is the number of spins in the simulator system, $\Pi_\alpha^{(i)}$ and $\Pi_\beta^{(i)}$ are the projectors onto the $\ket{\alpha}$ and $\ket{\beta}$ basis states applied to the $i^{\textrm{th}}$ spin, and $J$ is some large value which pushes all states which have non-zero energy with respect to $\Hpe$ to high energies.
To understand why this Hamiltonian gives a simulation of $\Htarget$ note that $\Hpe$ has a degenerate zero energy ground space, that is spanned by history states.
These history states encode the phase estimation algorithm applied to different states of the physical spins.
All other states have very high energies, due to the $J$ factor.
The projectors in $\Hsim$ break the ground space degeneracy of $\Hpe$ by giving the appropriate energy factors the $\ket{\alpha}$ and $\ket{\beta}$ spins which encode the output of the phase estimation algorithm.

It follows immediately from \cite[Lemma 4.3]{Kohler2020} that the Hamiltonian in \cref{eq:first order sim} is a $(\Delta,\epsilon,\eta)$-simulation of $\Htarget$ with: 
\begin{equation}\label{eq:eta}
\eta = O\left(\frac{T^3 \norm{\Htarget}}{J} + \frac{\Tpe}{T}\right)
\end{equation}
\begin{equation}
\epsilon = \epsilon' + O\left(\frac{T^4 \norm{\Htarget}^2}{J}\right)
\end{equation}
\begin{equation}
\Delta = \frac{J}{2T^2}
\end{equation}
where $J$ and $T$ are as defined in \cref{eq:first order sim}, $\Tpe$ is the time for the phase estimation part of the history state computation (i.e. $\Tpe = T-L$ for $L$ the number of idling steps), and $\epsilon'$ is the error in the phase estimation computation.

Therefore we will need to select $L$ so that $T \geq \frac{\Tpe}{\eta}$, giving:
\begin{equation}
J \geq \Delta T^2  + \frac{T^3 \norm{\Htarget}}{\eta} + \frac{T^4 \norm{\Htarget}}{\epsilon}.
\end{equation}

To compare this with the bound in \cref{them PBQC sim} we need to determine the time needed for the phase estimation algorithm.

\subsection{Quantum Phase Estimation computation time}

The QPE algorithm requires us to implement the unitary $U = e^{i \Htarget \tu}$ where $\norm{\Htarget \tu} \leq 2\pi$.
If we want to estimate the eigenvalues of $\Htarget$ to accuracy $\epsilon'$ we need to determine the phase of $U$ to accuracy $\frac{\epsilon'}{\norm{\Htarget}}$  which requires $\frac{\norm{\Htarget}}{\epsilon'}$ implementations of $U$. 

Using \cite{Berry:2015} we can implement a unitary $U$ generated by a $d$-sparse Hamiltonian which acts on $N$ spins to within accuracy $\epsilon''$ in time
\begin{equation}
T_{U} = O\left(\kappa \left(1 + N + \log^{\frac{5}{2}}\left(\frac{\kappa}{\epsilon''}\right)\right)    \frac{\log(\frac{\kappa}{\epsilon''})}{\log \log(\frac{\kappa}{\epsilon''})} \right)
\end{equation}
where $\kappa = d \norm{\Htarget}_{\max}\tu$.
Here $\norm{\Htarget}_{\max}$ denotes the largest entry of $\Htarget$ in absolute value.
We will choose $\tu$ so that $\norm{\Htarget}_{\max}\tu= O(1)$, and set  $\epsilon'' = \frac{\epsilon'}{\norm{\Htarget}}$ giving 
\begin{equation}
T_{U} = O\left(d \left(1 + N + \log^{\frac{5}{2}}\left(\frac{d \norm{\Htarget}}{\epsilon'}\right)\right)    \frac{\log(\frac{d\norm{\Htarget}}{\epsilon'})}{\log \log(\frac{d\norm{\Htarget}}{\epsilon'})} \right).
\end{equation}

Therefore we have that:
\begin{equation}\label{eqn TPE}
\Tpe = \frac{\norm{\Htarget}}{\epsilon'} T_{U} = O\left(\frac{d\norm{\Htarget}}{\epsilon'} \left(1 + N + \log^{\frac{5}{2}}\left(\frac{d\norm{\Htarget}}{\epsilon'}\right)\right)   
 \frac{\log(\frac{d\norm{\Htarget}}{\epsilon'})}{\log \log(\frac{d\norm{\Htarget}}{\epsilon'})} \right) .
\end{equation}
 
 \subsection{General simulation bounds}\label{app:sim bounds}

From \cref{eq:eta} we have that,
\begin{equation}
\eta \geq \frac{\Tpe}{T} = \frac{\Tpe}{\Tpe+L}.
\end{equation}
This follows since $\eta$ is at least as large as the second term in \cref{eq:eta}, and the total time for the computation is given by the time for the phase estimation algorithm plus the idling time.

So we can set $L = \lambda \Tpe$ for some large constant $\lambda$. 
This gives $T = O(\Tpe)$.

This gives:
\begin{equation}
\eta =  O\left(\frac{\norm{\Htarget}^4}{J}  \left(\frac{d}{\epsilon'} \left(N + \log^{\frac{5}{2}}\left(\frac{d\norm{\Htarget}}{\epsilon'}\right) \right)   \frac{\log(\frac{d\norm{\Htarget}}{\epsilon'})}{\log \log(\frac{d\norm{\Htarget}}{\epsilon'})} \right)^3 + \frac{1}{m}\right) 
\end{equation}
\begin{equation}
\epsilon =   O\left( \epsilon' +\frac{\norm{\Htarget}^{6}}{J}  \left(\frac{d}{\epsilon'} \left(N + \log^{\frac{5}{2}}\left(\frac{d\norm{\Htarget}}{\epsilon'}\right)\right)  \frac{\log(\frac{d\norm{\Htarget}}{\epsilon'})}{\log \log(\frac{d\norm{\Htarget}}{\epsilon'})} \right)^4 \right)
\end{equation}
\begin{equation}
\Delta= O\left( \frac{J}{\left(\frac{d \norm{\Htarget}}{\epsilon'} \left( N + \log^{\frac{5}{2}}\left(\frac{d\norm{\Htarget}}{\epsilon'}\right)\right)    \frac{\log(\frac{d\norm{\Htarget}}{\epsilon'})}{\log \log(\frac{d\norm{\Htarget}}{\epsilon'})} \right)^2 }   \right).
\end{equation}

Therefore:
\begin{equation}\label{eq:J}
J = \poly\left(\left(\frac{N^4d^4}{\epsilon^5}\right)^{\frac{1}{6}} \norm{\Htarget} ,\left(\frac{N^3d^3}{\eta \epsilon^3}\right)^{\frac{1}{4}} \norm{\Htarget},d^2 \Delta, \log(\frac{d\norm{\Htarget}}{\epsilon}), \frac{1}{\log(\log(\frac{d\norm{\Htarget}}{\epsilon}))} \right).
\end{equation}
where we have used that $\epsilon'$ is upper bounded by $\epsilon$. 
Therefore, for $d=O(1)$ this gives a simulation method where $\max_i{\norm{h_i}} = \left(\poly\left( \frac{N^a}{\epsilon^b} \norm{\Htarget}, \frac{N^x}{\epsilon^y \eta^z}\norm{\Htarget}\right)\right)$ where $a+2b = 2\frac{1}{3}$ and $x+2y +z = 2\frac{1}{2}$.
For polynomially sparse Hamiltonians we have that $a+2b > 2\frac{1}{3}$ and $x+2y +z> 2\frac{1}{2}$, where the tightness of the lower bounds depends on the degree of the polynomial in the sparsity. 
While these are close to the bounds from \cref{them PBQC sim} for the case of sparse Hamiltonians, for general Hamiltonians the $d$-factor in \cref{eq:J} contributes a factor which is exponential in $N$, and results in a scaling that is far from that in \cref{them PBQC sim}. 

\subsection{Holographic Quantum Error Correcting Codes using this construction}
In this section we investigate using these simulation techniques in the attack on QPV outlined in the proof of \cref{them PBQC sim} for general Hamiltonians. 
In this case we have $N = n\tau^{R}$.
The Pauli rank of the operator to be simulated as at most $4^n$ so $d=O(4^n)$.
Let us set $\norm{\Htarget} = \tau^{-\alpha R}4^{-\alpha n}$ and $\epsilon' = \tau^{-\beta R}4^{-\beta n}$ where we require $\beta \geq \alpha > 0$ to ensure that the accuracy with which the unitary $U = e^{i\Htarget \tu}$ is implemented is small with respect to $\norm{\Htarget}$.
We find:
\begin{equation}
\begin{multlined}
\Tpe = O \left(n4^{(\beta-\alpha + 1)n} \tau^{(\beta-\alpha)R}\left(n\tau^R + \left((\beta+\alpha+1)n \right)^{5/2} + \left((\beta-\alpha)R\log \tau \right)^{5/2}\right)\times\right.\\
\left. \frac{(\beta-\alpha+1)n+(\beta-\alpha)R\log\tau}{\log((\beta+1-\alpha)n+(\beta-\alpha)R\log\tau)} \right).
\end{multlined}
\end{equation}

For concreteness let us choose $\beta = \alpha$:
\begin{align}
\Tpe =& O \left(\frac{4^{n} (n^2\tau^{R} +n^{7/2})}{\log(n)} \right).
\end{align}

Giving:
\begin{equation}
\eta = O \left(\frac{4^{(3-\alpha)n}\tau^{-\alpha R}(n^6 \tau^{3R}+n^{21/2})}{J \log^3(n)}\right)
\end{equation}
\begin{equation}
\epsilon= O \left(\frac{4^{2(2-\alpha)n}\tau^{-2\alpha R}(n^8\tau^{4R}+n^14)}{J \log^4(n)}\right)
\end{equation}
\begin{equation}
\Delta =O\left( \frac{J \log^2(n)}{4^{2n}(n^2\tau^{2R}+n^7)} \right).
\end{equation}

Therefore if we set $J=O(1)$ as required for causality in the toy model, then we can construct a `good' simulation in terms of $\epsilon$ and $\Delta$ (i.e. a simulation where $\epsilon \ll \norm{\Htarget}$ and $\Delta \gg \norm{\Htarget}$) if we choose $\alpha >4$.
For concreteness we will set $\alpha=5$. 
This gives:
\begin{equation}
\eta = O \left(\frac{4^{-2n}\tau^{-5R}(n^6\tau^{3R}+n^{21/2})}{\log^3(n)}\right)
\end{equation}
\begin{equation}
\epsilon= O \left(\frac{4^{6n}\tau^{-10R}(n^8\tau^{4R}+n^{14})}{\log^4(n)}\right)
\end{equation}
\begin{equation}
\Delta =O\left( \frac{\log^2(n)}{4^{2n}(n^2\tau^{2R}+n^7)} \right).
\end{equation}

However, recall that in order to maintain the link to QPV we can only evolve the system for time $\frac{1}{4}Cn\tau^R$, so in this set-up we can only perform a unitary with ``computation budget'' at most $\norm{\Hbulk}t = n 4^{-5n}\tau^{-4R}$. 
From \cref{physical-properties} the error in the state at the end of the simulation would be: 
\begin{equation}
\begin{multlined}
      \|\ee^{-\ii \Hsim t}\rho'\ee^{\ii \Hsim t} - \ee^{-\ii \mathcal{E}(\Htarget)t}\rho'\ee^{\ii \mathcal{E}(\Htarget)t}\|_1 \\
      \leq O\left(\frac{4^{6n}\tau^{-9R}(n^7 \tau^{3R} + n^{23/2})}{\log^4(n)}+ \frac{4^{-2n}\tau^{-5R}(n^6\tau^{3R}+n^{21/2})}{\log^3(n)}\right).
      \end{multlined}
      \end{equation}
      
We see that while the $\epsilon$-error is small (and can be made arbitrarily so by reducing the norm of $\Htarget$), we cannot achieve a error in $\eta$ which is small with respect to the computation budget without increasing $J$.
While the $\eta$ error does not accumulate with time it nevertheless dominates the error for the small computation budget we are considering here, and this error swamps the simulation.

\subsection{Holographic Quantum Error Correcting Codes using this construction for a restricted class of bulk Hamiltonians}

In this section we will examine the results of applying the simulation techniques from \cref{app:sim bounds} in a HQECC set up for bulk Hamiltonians which are $k$-local for constant $k$.
This guarantees that the target Hamiltonian on the boundary of the HQECC is sparse, therefore from \cref{app:sim bounds} we know that we can get very close to saturating the definition of very good simulators.
However, in this section we will see that despite getting very close to the bounds, the simulations do not provide interesting attacks on position based verification because the errors in the simulation dominate the computation being carried out.

We still have that $N = n\tau^{R}$, and now we have $d=O(1)$.
Let us set $\norm{\Htarget} = \tau^{-\alpha R}n^{-\alpha' k}$ and $\epsilon' = \tau^{-\beta R}n^{-\beta' k}$ where we require $\beta \geq \alpha > 0$ and $\beta' \geq \alpha' > 0$ to ensure that the accuracy with which the unitary $U = e^{i\Htarget \tu}$ is implemented is small with respect to $\norm{\Htarget}$.
Substituting the above into \cref{eqn TPE} we find:
\begin{equation}
\begin{multlined}
\Tpe = O \left(n^{(\beta' -\alpha'+1)k + 1} \tau^{(\beta-\alpha)R}\left(n\tau^R + \left((\beta'-\alpha'+1)k\log n \right)^{5/2}+  \left((\beta-\alpha)k\log \tau \right)^{5/2} \right) \times\right.\\
\left.\frac{(\beta'-\alpha'+1)\log(n)+(\beta-\alpha)R\log(\tau)}{\log((\beta'-\alpha'+1)\log(n)+(\beta-\alpha)R\log(\tau))} \right).
\end{multlined}
\end{equation}

For concreteness let us choose $\beta = \alpha$ and $\beta' = \alpha'$:
\begin{align}
\Tpe &= O \left(\frac{n^{k}(n\tau^R + \log^{5/2}n)\log(n) }{\log\log(n)} \right)\\
& =O \left(\frac{n^{k+1}\tau^R \log(n) }{\log\log(n)} \right) .
\end{align}

\cref{eq:J} now implies the following scalings on the simulation parameters:
\begin{equation}
\eta = O \left(\frac{n^{(3-\alpha')k + 3}\tau^{(3-\alpha)R}\log^3(n)}{J \log^3(\log(n))}\right)
\end{equation}
\begin{equation}
\epsilon= O \left(\frac{n^{2(2-\alpha')k + 4}\tau^{2(2-\alpha)R}\log^4(n)}{J \log^4(\log(n))}\right)
\end{equation}
\begin{equation}
\Delta =O\left( \frac{J\log^2(\log(n))}{n^{2(k+1)}\tau^{2R}\log^2(n)} \right).
\end{equation}

Therefore if we set $J=O(1)$ as required for causality in the toy model, and assume $k=2$ for concreteness, then we can construct a `good' simulation in terms of $\epsilon$ and $\Delta$ (i.e. a simulation where $\epsilon \ll \norm{\Htarget}$ and $\Delta \gg \norm{\Htarget}$) if we choose $\alpha' >6$ and $\alpha > 4$.
For concreteness we will set $\alpha'=7$ and $\alpha=5$. 
This gives:
\begin{equation}
\eta = O \left(\frac{n^{-5}\tau^{-2R}\log^3(n)}{\log^3(\log(n))}\right)
\end{equation}
\begin{equation}
\epsilon= O \left(\frac{n^{-16}\tau^{-6R}\log^4(n)}{\log^4(\log(n))}\right)
\end{equation}
\begin{equation}
\Delta =O\left( \frac{\log^2(\log(n))}{n^6 \tau^{2R}\log^2(n)} \right).
\end{equation}

However, recall that in order to maintain the link to QPV we can only evolve the system for time $\frac{1}{4}Cn\tau^R$, so in this set-up we can only perform a unitary with ``computation budget'' at most $\norm{\Hbulk}t = n^{-13}\tau^{-4R}$. 
From \cref{physical-properties} the error in the state at the end of the simulation would be: 
\begin{equation}
\begin{multlined}
      \|\ee^{-\ii \Hsim t}\rho'\ee^{\ii \Hsim t} - \ee^{-\ii \mathcal{E}(\Htarget)t}\rho'\ee^{\ii \mathcal{E}(\Htarget)t}\|_1 \\
      \leq O\left(\frac{n^{-15}\tau^{-5R}\log^4(n)}{\log^4(\log(n))} +  \frac{n^{-5}\tau^{-2R}\log^3(n)}{\log^3(\log(n))}\right).
      \end{multlined}
      \end{equation}
      
As in the previous section the $\epsilon$ error is small with respect to the computation budget (and can be made arbitrarily smaller by reducing $\Htarget$), however we cannot achieve a error in $\eta$ which is small with respect to the computation budget without increasing $J$.
While the $\eta$ error does not accumulate with time it nevertheless dominates the error for the small computation budget we are considering here, and this error swamps the simulation.

\section{Explicit perturbative construction of an attack on QPV}\label{sect pert const}

In this section we investigate the limits on the computational budget for the local unitary in the centre of the bulk, if perturbative simulation techniques are used to create an approximate boundary Hamiltonian that can implement the bulk computation non-locally. 

Using perturbative simulations gives worse bounds than the non-perturbative techniques described above, particularly in the case of sparse Hamiltonians.
However the perturbative techniques are not optimised for the cases where we are simulating Hamiltonians that have large locality but are sparse. 
It could be expected that perturbative simulations would perform worse since previous perturbative simulations generically require  $\max_i{\norm{h_i}} = O\left(\norm{H_{\textup{target}}}\textup{poly}\left(\frac{1}{\epsilon^k}\right)\right)$ to simulate a $k$-local $H_{\textup{target}}$ with a 2-local $H = \sum_i h_i$ while keeping the error and the number of spins fixed \cite{Cao2017,Bravyi2008,Cao2018}.
In the following section we give an explicit recipe for a perturbative simulation of the boundary Hamiltonian.
The interaction strengths of the simulator Hamiltonian is improved over the recipe followed in \cite{Kohler2019,Apel2021} by varying the `heavy' Hamiltonian norm\footnote{Perturbative simulations consist of a `heavy' Hamiltonian and a perturbation. When projected into the low energy subspace the heavy Hamiltonian ensures that the ancillary qubits are in a low energy state that with the perturbation facilitates the more complex interaction being simulated.}.

We want to simulate a $n$-local bulk Hamiltonian term acting on the central bulk tensor (which generates the general $n$ qubit unitary).
This has Pauli rank at most $4^n$.
The tensor network mapping gives a boundary Hamiltonian, $H_\textup{target}$, that is a perfect simulation of this bulk operator.
The norm and Pauli rank are preserved since the mapping is isometric and stabilizer, however, it is now $O(n\tau^R)$-local.  
The initial target Hamiltonian on the boundary before any simulation is therefore,
\begin{equation}
H_\textup{target} = b P_{O(n\tau^R)},
\end{equation}
where $b$ is a placeholder for the norm of the bulk term which we will later examine.

To recover a causal boundary Hamiltonian we construct an approximate simulation to $H_\textup{target}$ with reasonable interaction strengths that is $2$-local.
This section uses perturbative simulation techniques originating in complexity theory.
Reducing the locality employs the subdivision gadget \cref{sect: subdivision}.
Simulating a $n\tau^R$-local Hamiltonian with a $2$-local simulator requires $O(R\log(n))$ recursive applications of the subdivision gadget followed by the 3-to-2 gadget.

Since the initial interaction had Pauli rank $4^n$, this many recursive simulations are done in parallel to break down the full initial operator.
Hence the entire process introduces $O(4^n n \tau^R)$ ancilla qubits.
Once the locality is reduced there are additional gadgets to make the boundary Hamiltonian geometrically local, see proof of Theorem 6.10 from \cite{Kohler2019}.
However, reducing the locality (to a 2-local but not geometrically local Hamiltonian) limits the computational budget to doubly exponential in $R$ and exponential in $n$.

\subsection{Perturbative simulation tools}\label{sect: pert appen}

Here for completeness we quote some relevant technical results concerning Hamiltonian simulation from the literature.
\begin{lemma}[Second order simulation~\cite{Bravyi2014}] \label{lm: second_order} Let $V = H_1 + \Delta^{\frac{1}{2}}H_2$ be a perturbation acting on the same space as $H_0$ such that $\max(||H_1||,||H_2||) \leq \Lambda$; $H_1$ is block diagonal with respect to the split $\mathcal{H} = \mathcal{H}_- \oplus \mathcal{H}_+$ and $(H_2)_{--} = 0$.
  Suppose there exists an isometry $W$ such that $\textup{\emph{Im}}(W) = \mathcal{H}_-$ and:
  \begin{equation} \label{second_order_eq}
    || W H_{\textup{\emph{target}}} W^\dagger - (H_1)_{--} + (H_2)_{-+}H_0^{-1}(H_2)_{+-}||_{\infty} \leq \frac{\epsilon}{2}
  \end{equation} Then $\tilde{H} = H + V$ $(\frac{\Delta}{2}, \eta, \epsilon)$ simulates $H_{\textup{\emph{target}}}$, provided that $\Delta \geq O(\frac{\Lambda^6}{\epsilon^2} + \frac{\Lambda^2}{\eta^2})$.
\end{lemma}

\begin{lemma}[Concatenation of approximate simulation~\cite{Cubitt2019}] \label{lm: concat} 
Let $A$, $B$, $C$ be Hamiltonians such that $A$ is a $(\Delta_A,\eta_A,\epsilon_A)$-simulation of $B$ and $B$ is a $(\Delta_B,\eta_B,\epsilon_B)$- simulation of $C$.
Suppose $\epsilon_A,\epsilon_B \leq \norm{C}$ and $\Delta_B\geq \norm{C} + 2 \epsilon_A + \epsilon_B$.
Then $A$ is a $(\Delta,\eta,\epsilon)$-simulation of $C$ where $\Delta \geq \Delta_B - \epsilon_A$ and 
\begin{equation}
\eta = \eta_A + \eta_B + O \left(\frac{\epsilon_A}{\Delta_B - \norm{C}+\epsilon_B} \right)
\quad and \quad
\epsilon = \epsilon_A + \epsilon_B + O \left(\frac{\epsilon_A\norm{C}}{\Delta_B - \norm{C} + \epsilon_B} \right).
\end{equation}
\end{lemma}

In \cref{sect: pert scaling} we use the following technical result.
\begin{lemma}[Closed form of $\delta_x$] \label{lm delta_x} 
And element in the series $\delta$ defined recursively,
\begin{equation}
\delta_x = \begin{cases}
\delta_0 & x=0\\
4\sum_{l=0}^{x-1}\frac{\delta_l}{2^{x-l}} & x>0,
\end{cases}
\end{equation}
has the closed form for $x>0$,
\begin{equation}
\delta_x = \frac{5^{x-1}}{2^{x-2}}\delta_0.
\end{equation}
\end{lemma}
\begin{proof}
The first few elements of the series read,
\begin{align}
\delta_1 & = 4\left[\frac{1}{2} \right]\delta_0\\
\delta_2 & = 4\left[\frac{1}{2^2} + \frac{4}{2^2} \right]\delta_0\\
\delta_3 & = 4\left[\frac{1}{2^3} + \frac{2\times 4}{2^3} + \frac{4^2}{2^3} \right]\delta_0\\
\delta_4 & = 4\left[\frac{1}{2^4} + \frac{3 \times 4}{2^4} + \frac{3 \times 4^2}{2^4} + \frac{4^3}{2^4} \right]\delta_0\\
\delta_5 & = 4\left[\frac{1}{2^5} + \frac{4\times 4}{2^5} + \frac{6 \times 4^2}{2^5} + \frac{4 \times 4^3}{2^5}+ \frac{4^4}{2^5}\right]\delta_0.
\end{align}
We notice the pattern,
\begin{equation}
\delta_x = \frac{2^2}{2^x} \sum_{k=0}^{x-1} {x-1\choose k} 4^{k}.
\end{equation}
We then use the identity $\sum_{k=0}^{m} {m\choose k} 4^{k} = 5^{m}$ to obtain the quoted result. 
\end{proof}

\subsubsection{Subdivision gadget used in the construction}\label{sect: subdivision}
The subdivision gadget is used to simulate a $k$-local interaction by interactions which are at most $\ceil{\frac{k}{2}}+1$-local.
We want to simulate the Hamiltonian:
\begin{equation}
H_{\textup{target}} = \delta_0 (P_{A}\otimes P_{B} + P_{A}^\dagger\otimes P_{B}^\dagger)
\end{equation}
Let $\tilde{H} = H+V$ where:
\begin{equation}
H = \Delta \Pi_+
\end{equation}
\begin{equation}
V = H_1 + \Delta^{\frac{1}{2}}H_2
\end{equation}
where:
\begin{equation}
\Pi_+ = \ket{1}\bra{1}_w+\ket{2}\bra{2}_w+...+\ket{p-1}\bra{p-1}_w
\end{equation}
\begin{equation}
H_1 = 2\delta_0 \identity 
\end{equation}
\begin{equation}
H_2 =\sqrt{ \frac{\delta_0}{2}}\left(-P_{A}\otimes X_w - P_A^\dagger \otimes X_w^\dagger + P_B\otimes X_w^\dagger + P_B^\dagger\otimes X_w\right)
\end{equation}
The degenerate ground space of $H$ has the mediator qubit $w$ in the state $\ket{0}\bra{0}$ so $\Pi_- = \ket{0}\bra{0}_w$. This gives:

\begin{equation}
(H_{1})_{--}=  2 \delta_0 \identity  \otimes \ket{0}\bra{0}_w
\end{equation}
and:
\begin{equation}
(H_2)_{-+} =\sqrt{ \frac{\delta_0}{2}}\left(-P_A \otimes \ket{0}\bra{p-1}_w - P_A^\dagger \otimes \ket{0}\bra{1}_w + P_B \otimes \ket{0}\bra{1}_w + P_B^\dagger \otimes \ket{0}\bra{p-1}_w\right)
\end{equation}
Therefore:
\begin{equation}
(H_2)_{-+}H_0^{-1}(H_2)_{+-} =\delta_0 \left(P_A \otimes P_B + P_A^\dagger \otimes P_B^\dagger - 2\identity  \right) \otimes \ket{0}\bra{0}_w
\end{equation}

If we define an isometry $W$ by $W\ket{\psi}_A = \ket{\psi}_A\ket{0}_w$ then:
\begin{equation}
 || W H_{\textup{target}} W^\dagger - (H_1)_{--} + (H_2)_{-+}H_0^{-1}(H_2)_{+-}|| =  0
\end{equation}
Therefore \cref{second_order_eq} is satisfied for all $\epsilon \geq 0$. So, provided a $\Delta$ is picked which satisfies the conditions of \cref{second_order_eq}, $\tilde{H}$ is a $(\frac{\Delta}{2},\eta,\epsilon)$-simulation of $H_{\textup{target}}$.

\subsection{Limited computational budget with perturbative simulation}\label{sect: pert scaling}

At round $r$ we use a "heavy" Hamiltonian of strength,
\begin{equation}
\Delta_r = b \gamma_r^{-2},
\end{equation}
where $\gamma_r$ is initially general and we just write the above for algebraic convenience that will become apparent. 

Using second order $(\epsilon, \eta, \Delta)$-simulation the errors (\cref{lm: second_order}) in each round is given by,
\begin{equation} 
\epsilon_r \geq O \left( \frac{\Lambda_r^3}{\sqrt{\Delta_r}}\right), \qquad \eta_r \geq O\left( \frac{\Lambda_r}{\sqrt{\Delta_r}}\right)
\end{equation}
where $\Lambda = \max \{ \norm{H_1^{(r)}}, \norm{H_2^{(r)}}\}$.

We will denote by $h_t^{(r)}$ the norm of the Hamiltonian of the target Hamiltonian that is simulated in round $r$, and assume $b<1$ so that $\sqrt{b}>b$.
After enumerating a few rounds of the subdivision gadget (\cref{sect: subdivision}) the general scaling of the simulation error, $\epsilon_r$, and the Hamiltonian norm as the subdivision is applied recursively is found to be 
\begin{align}
h_t^{(r)} & = O \left(b \gamma_r^{-1}\gamma_{r-1}^{-1/2}...\gamma_0^{-1/2^r} \right),\label{eqn ht}\\
\epsilon_r &\geq O \left(b \frac{\gamma_r}{\left(\gamma_0^{1/2^r}...\gamma_{r-1}^{1/2}\right)^3}  \right),\label{eqn epsilonr}\\
\eta_r & \geq O\left(\frac{\gamma_r}{\gamma_{r-1}^{1/2^2}\gamma_{r-2}^{1/2^3}...\gamma_0^{1/2^{r+1}}}\right),
\end{align}
for $r>1$.
The convenience of working with $\gamma_r$ rather that $\Delta_r$ becomes apparent since the error and the Hamiltonian are always proportional to $b$.

Now to turn the products into sums define $\delta_r$ by,
\begin{equation}
\gamma_r^{-2} = \tau^{\delta_r \log(n) R},
\end{equation}
so that now \cref{eqn ht} and \cref{eqn epsilonr} read as:
\begin{align}
h_t^{(r+1)} & = O \left(b \tau^{R\log(n)\sum_{l=0}^{r} \frac{\delta_l}{2^{r-l}}} \right)\label{eqn ht2}\\
\epsilon_r &\geq O \left(b \tau^{-\frac{R\log(n)}{2}\left(\delta_r - 3 \sum_{l=0}^{r-1}\frac{\delta_l}{2^{r-l}}\right)}  \right)\label{eqn epsilonr2}\\
\eta_r & \geq O \left(\tau^{-\frac{R\log(n)}{2}}\left(\delta_r - \frac{1}{4}\sum_{l=0^{r-1}}\frac{\delta_l}{2^{r-l}} \right) \right)
\end{align}
We can now analyse what we require to have a good simulation. 

From \cref{lm: concat} (if the conditions hold) the final error in the Hamiltonian is given by,
\begin{equation}
\mathbf{\epsilon} = \sum_{r=0}^{\log(n)R} \epsilon_r, \qquad \mathbf{\eta} =  \sum_{r=0}^{\log(n)R} \eta_r.
\end{equation}
For error in the simulation to be small with respect to the target Hamiltonian norm $b$, we require,
\begin{equation}
\delta_r > 3 \sum_{l=0}^{r-1}\frac{\delta_l}{2^{r-l}}.
\end{equation}
To ensure this we set $\delta_r = (3+a) \sum_{l=0}^{r-1}\frac{\delta_l}{2^{r-l}}$ for some constant $a$.
This choice also ensures that the conditions in \cref{lm: concat} are satisfied and we 
The choice of $a$ does not affect the scaling so we set $a=1$ since \cref{lm delta_x} gives a convenient closed form for $\delta_r$ (see \cref{sect: pert appen}).

The sum in \cref{eqn ht2} is now given by,
\begin{align}
\sum_{l=0}^{r} \frac{\delta_l}{2^{r-l}}&=\frac{\delta_0}{2^r} + \sum_{l=1}^{r} \frac{1}{2^{r-l}}\frac{5^{l-1}}{2^{l-2}}\delta_0 \\
& = \frac{\delta_0}{2^r} +\frac{\delta_0}{2^{r-2}}\sum_{l=1}^{r}5^{l-1}\\
& =  \frac{\delta_0}{2^r} + \delta_0 2^{-r} (5^r - 1)\\
& = \left(\frac{5}{2}\right)^r \delta_0
\end{align}
Hence,
\begin{align}
h_t^{(r+1)} & = O \left(b \tau^{R\log(n)(5/2)^r\delta_0} \right)\label{eqn ht3}\\
\epsilon_r &\geq O \left(b \tau^{-\frac{R\log(n)}{2}\left(\frac{5}{2}\right)^{r-1}\delta_0}  \right)\label{eqn epsilonr3}\\
\eta_r & \geq O \left(\tau^{-\frac{R\log(n)}{4}\left(\frac{5}{2}\right)^{r-1}\delta_0} \right)
\end{align}

The final approximate simulation has,
\begin{align}\label{eqn sim scaling}
&\epsilon  \geq O \left(\log(n)Rb\tau^{-R\log(n)\delta_0/4}  \right)\\
& \eta \geq O \left( \log(n)R \tau^{-\log(n)R\delta_0/2} \right)\\
& \Delta  \geq \Delta_1 - \epsilon' = b\tau^{\delta_0 R \log(n)} - O(\log(n)Rb \tau^{-R\log(n)\delta_0/4})
\end{align}
since the first round of simulation has the largest upper bound of \\
${\epsilon_1 \geq O\left(b\tau^{-R\log(n)\delta_0/4} \right)}$ and $\eta_1 \geq O \left(\tau^{-\log(n)R\delta_0/2} \right)$.
The final simulating Hamiltonian is norm,
\begin{equation}\label{eqn sim norm}
\norm{H_\textup{sim}} = O(b\tau^{Rn\log(n)(5/2)^R\delta_0}).
\end{equation}

We require two conditions such that the boundary Hamiltonian is an approximate simulation of the central bulk operation that is causal. 
Firstly, the simulation Hamiltonian must have small interaction strengths so that the Lieb Robinson is constant, i.e. \cref{eqn sim norm} must be at most $O(1)$ with respect to $R$.
Secondly the simulation approximation must be small with respect to the target Hamiltonian norm, i.e. from \cref{eqn sim scaling}, $O(\log(n)R\tau^{-R\log(n)\delta_0/4}) = O(1)$.
In order to do this we are free to chose $\delta_0$.

To achieve a reasonable Lieb-Robinson velocity it would be helpful to make $\delta_0$ scale as $O\left(\frac{e^{-R}}{nR}\right)$, this would obtain $O(1)$ Lieb-Robinson velocities for $b=O(1)$ norm bulk Hamiltonians.
However this leads to a simulation error that completely swamps the target Hamiltonian in \cref{eqn sim scaling}.
Imposing a reasonable simulation error restricts $\delta_0$ to be a small constant independent of $R$ and $n$. 
Once $\delta_0$ is set, to achieve a reasonable Lieb-Robinson velocity from $\norm{H_\textup{sim}}$ the bulk Hamiltonian norm is forced to be doubly exponentially small in $R$ and singularly exponentially small in $n$: $b = O\left(\tau^{-nR\log(n)(5/2)^R} \right)$.
Hence the complexity of the local unitary is of very low complexity and the protocol is trivial.

We have argued that with perturbation gadgets we do not expect to be able to do any better than this limited protocol.

\section{Lower bound on simulation from causality}\label{appen: LR}

Here we consider a simple causality argument to get a lower bound on the strengths of a 2-geometrically local Hamiltonian that simulates a $n$ local one. 
Consider the target Hamiltonian $H  = \frac{1}{2}(\identity_A \otimes \identity_B + X_A\otimes X_B + Y_A\otimes Y_B + Z_A\otimes Z_B)$.
While this Hamiltonian is $2$-local we consider $A$ and $B$ being $n$ qubits apart so that the Hamiltonian spreads correlations across $n$ qubits almost instantaneously (clearly a feature of $n$-local systems).
Given an initial state $\ket{\psi}_A\otimes \ket{0}_B$ evolution under this Hamiltonian for a time $t$ is given by, 
\begin{equation}
\ee^{-\ii tH}\ket{\psi}_A \ket{0}_B =\cos(t) \ket{\psi}_A\ket{0}_B - i \sin(t) \ket{0}_A\ket{\psi}_B.
\end{equation}
Hence in time $\pi/2$ the information at $A$ has been transferred to $B$ under the Hamiltonian.

If we now simulate this Hamiltonian with $\tilde{H} = \sum_{i=1}^{n-1}\tilde{h}_{i,i+1}$ where $\tilde{h}_{i,i+1} =  \frac{\tau}{2}(\identity_i \otimes \identity_{i+1} + X_i \otimes X_{i+1} + Y_i\otimes Y_{i+1} + Z_i\otimes Z_{i+1})$.
Given the initial state $\ket{\psi}_1\ket{0}_n$, we impose the state after evolution for $t=\pi/2$ must be $\ket{0}_1\ket{\psi}_n$ to emulate the target Hamiltonian.
The interaction strength of the simulator Hamiltonian must be $\tau = n$ in order to propagate the information through the chain of interactions in the same time. 
Hence this gives an example of a sparse $n$-local Hamiltonian that a universal geometrically local Hamiltonian will have to simulate which requires interaction strength linear in $n$.
While we give a specific example of a simulating Hamiltonian it is clear that the form of the interactions cannot be modified to propagate information faster since each link is maximally entangling. 

Consider an approximation in the time dynamics simulator, $\norm{\ee^{-\ii \tilde{H} \pi/2}\ket{\psi}_1\ket{0}_n - \ket{0}_n\ket{\psi}_1}_1 \leq \epsilon$.
As in \cref{them PBQC sim} we may hope to find a lower bound on the interaction strength that encorportates the error in the simulation. 
For this example we can translate the error in the state to, 
\begin{equation}
\abs{\cos^n\left(\frac{\tau \pi}{2 n} \right)}^2 \leq \epsilon. 
\end{equation}
We can see here if $\tau$ scales sub-linearly with $n$ then for large $n$ we have $\abs{\cos^n\left(\frac{\tau \pi}{2 n} \right)}^2 = 1$ and the simulator state is orthogonal to the target. 
So in this case we are unable to obtain an interaction strength bound incorporating more parameters of the simulation.

This investigation, while rudimentary, intuitively demonstrates the lower bounds presented in this paper are not immediate from causality arguments. 
We leave further investigation into other methods of reasoning about simulation lower bounds to further work.

\end{appendices}

\printbibliography

\end{document}